\documentclass[aps,superscriptaddress,nofootinbib,a4paper,longbibliography,twocolumn]{revtex4-2}
\usepackage{amsmath}
\usepackage{amsfonts}
\usepackage{amsthm}
\usepackage{amssymb}
\usepackage{graphicx}
\usepackage{enumerate}
\usepackage{color}
\usepackage[T1]{fontenc}
\usepackage{braket}
\usepackage{todonotes}
\usepackage{soul}
\usepackage{subfigure}
\usepackage{mathrsfs}
\usepackage{float}
\usepackage{amsmath}
\usepackage{bm}
\usepackage{multirow}
\usepackage{tabularx}
\usepackage{outlines}

\graphicspath{{figures/}}

\newtheorem{theorem}{Theorem}
\newtheorem{lemma}{Lemma}

\newtheorem*{theorem*}{Theorem}
\newtheorem{definition}{Definition}

\newtheorem*{rem}{Remark}
\newtheorem{prop}{Proposition}

\newtheoremstyle{operation}
  {\topsep}
  {\topsep}
  {}
  {}
  {\itshape}
  {.}
  {.5em}
  {\thmname{#1}\thmnumber{ #2}\thmnote{ (#3)}}
\theoremstyle{operation}

\def\>{\rangle}
\def\<{\langle}

\newcounter{protocol}
\renewcommand{\theprotocol}{\arabic{protocol}}

\usepackage[bookmarks=false,colorlinks,citecolor=blue,
linkcolor=blue,anchorcolor=blue,urlcolor=blue
]{hyperref}

\begin{document}
\title{Strong $O$-valued contextuality: ruling out discrete nondeterministic alternatives to quantum theory}
\author{Ravishankar Ramanathan}
\affiliation{School of Computing and Data Science, The University of Hong Kong, Pokfulam Road, Hong Kong}

\begin{abstract}
Gleason's theorem identifies the Born rule via non-contextuality over an infinite continuous lattice of projectors, while its corollary the Kochen-Specker (KS) theorem rules out the specific class of deterministic ($\{0,1\}$-valued) noncontextual models using finite sets of projectors. Gleason's theorem also indicates the existence of KS-type finite vector constructions to rule out other discrete nondeterministic alternatives to quantum theory beyond the $\{0, 1\}$ case. Here, we construct finite measurement configurations to rule out noncontextual empirical models with outcome probabilities drawn from an arbitrary finite subset $O \subset [0,1]$. We do this by two means: (i) constructing a family of experimentally feasible state-dependent Hardy-type tests, and (ii) proving a generalized KS theorem for a broad class of $O$ that includes prior results as special cases. In the sheaf-theoretic framework of Abramsky and Brandenburger, a hierarchy of probabilistic-possibilistic-strong contextuality has been established quantifying contextuality as a resource. We extend this framework by introducing strong $O$-valued contextuality, showing that quantum theory evades global sections of all finite-valued presheaves. We also discuss the implications of the result on finite many-valued logics as viable ontological models for quantum theory and for contextuality-based (semi)-device-independent protocols.

\end{abstract}
\maketitle

\textit{Introduction.-} Contextuality, as evidenced by the Kochen-Specker (KS) theorem \cite{KS67}, is a foundational property asserting that classical theories are insufficient to explain quantum correlations. The KS theorem explicitly shows that unlike in classical physics where properties exist independently of the choice to observe them, quantum systems do not possess predetermined values that are simply revealed by experiments. Since the advent of quantum information science, contextuality has gone beyond being treated as a philosophical quirk and is now recognized as a resource in several contexts - it provides the magic behind the speed up in fault-tolerant computing via magic state distillation \cite{HWVE14}, it has been shown to be a backbone in (semi)-device-independent quantum cryptography \cite{LRH+23, LR25}, in self-testing quantum states \cite{BRV+19} etc. Owing to its high interest, simple and minimal proofs of contextuality have been sought, and records for KS proofs have been established in small dimensions. 

To elaborate, the KS theorem \cite{KS67} states that for every quantum system belonging to a Hilbert space of dimension $d \geq 3$, irrespective of its actual state, a finite set of measurements exists that does not admit a deterministic non-contextual outcome assignment. Here, deterministic indicates that the outcome probabilities only take values in $\{0,1\}$. Non-contextual means that the assignment is made to the individual projectors, independently of the context (the other projectors in the measurement basis) to which they belong. As such the KS theorem is a powerful result ruling out empirical models of nature that obey the twin conditions of (i) \emph{Non-Contextuality} and (ii) \emph{Determinism}. 

A natural question is to determine to what extent these two assumptions can be relaxed for a hidden variable theory, while still being unable to explain quantum theory. Specifically, in view of the aforementioned applications and due to the fact that in practical experiments one cannot be certain that the exact same projector is being measured in multiple contexts, methods have been devised to test for $\epsilon$-relaxed noncontextual models \cite{KDL15, Fyrillas23, LR25}. On the other hand, experimental tests ruling out nondeterministic noncontextual alternatives to quantum theory are less studied in the literature \cite{Pitowsky98, HP04, Ravi24}. In this regard, it is well-known that the KS theorem is the finite analog of the powerful Gleason's theorem \cite{Gleason57}, and can indeed be seen as its corollary by means of a logical compactness argument \cite{Enderton01} (see App. \ref{sec:Gleason} for a proof). Gleason's theorem establishes the Born rule as being the only consistent non-contextual probability assignment to the entire infinite lattice of projections for systems of dimension $d \geq 3$. Here, a non-contextual probability assignment to a set of vectors $|v_i\rangle$ (equivalently the projectors $P_{v_i} = |v_i \rangle \langle v_i|$) is an assignment of the form $p: |v_i \rangle \rightarrow [0,1]$, where the probabilities are assigned directly to the projectors independent of the specific context (the other commuting projectors) in which they are measured. Gleason's theorem then shows that when one considers the entire set of projectors in Hilbert spaces of dimension $d \geq 3$, the only consistent assignment satisfying the condition that the sum of probabilities for any orthonormal basis equals one is of the form $\text{Pr}(v) = \text{Tr}\left[\rho P_v \right]$ for some density operator $\rho$. As such, Gleason's theorem derives the Born rule as a logical consequence of the quantum formalism, and as a corollary establishes the KS theorem in a non-constructive manner. However, Gleason's theorem is not deployed in practical applications since it applies to an infinite number of projection operators and its proofs rely on the continuity of functions on a compact sphere. It is desirable therefore, to pursue finite constructive proofs, that would give rise to a strengthening of contextuality as a resource in quantum information science and to identify the corresponding applications. Indeed, Pitowsky \cite{Pitowsky98} first raised the question of providing finite constructive proofs to rule out outcome probability assignments beyond the $\{0,1\}$ pointing out that the "general constructive case is, to the best of my knowledge, open" \cite{Pitowsky98}. 

Motivated by the above considerations, in \cite{Ravi24}, we proved a generalisation of the Kochen-Specker theorem to rule out, in dimension $d =3$, noncontextual nondeterministic outcome probability assignments of the form $p: |v_i \rangle \rightarrow \{0,q,1-q,1\}$ for $0 \leq q \leq 1/2$ with $q \neq 1/3$. This strengthens the KS theorem since it rules out a larger set of noncontextual outcome probability assignments (with values in $\{0,q,1-q,1\}$ and their convex mixtures rather than $\{0,1\}$ and their convex mixtures). While finite and constructive, the number of measurements required for this proof was still too high to be feasibly tested in current experiments. In this paper, we first solve the problem in generality in a state-dependent manner. Formally, we present a family of experimentally feasible Hardy-type tests to rule out general nondeterministic noncontextual alternatives to quantum theory, ruling out noncontextual probability assignments from any finite set $O \subset [0,1]$. Secondly, we present a strengthening of the result from \cite{Ravi24} by proving a generalized KS theorem for a broad class of $O$ that includes the previous results as a special case. Contextuality is sought after as a resource in aforementioned quantum information applications, a rigorous means to quantify contextuality as a resource was established in \cite{AB11, ABKLM15, HZS+23, ABM17}. Specifically in the sheaf-theoretic framework of contextuality instituted by Abramsky and Brandenburger \cite{AB11}, a hierarchy of probabilistic-possibilistic-strong contextuality has been established. We extend this framework by introducing strong $O$-valued contextuality as an even stronger resource, showing that quantum theory evades global sections of all finite-valued presheaves. We then point out a connection to quantum logic and its established isomorphism with infinite many-valued {\L}ukasiewicz logics. As a corollary, we prove that our constructions rule out the corresponding finite many-valued logics as viable ontological models for quantum theory. Finally, we discuss the implications of the result for contextuality-based quantum protocols for randomness generation. Specifically, we show that the optimal quantum correlations in our tests get increasingly close to an extremal nonclassical point of the general no-disturbance polytope as the number of measurements $6k+1$ grows. This provides a pathway towards achieving security of the corresponding semi-device-independent protocol for randomness expansion and amplification \cite{LR25} against the most powerful adversaries for this task, an application which we explicitly formulate in forthcoming work.

\textit{Ruling out nondeterministic noncontextual alternatives to quantum theory - generalising the KS theorem.-}
Hardy proofs have been noted for their simplicity, indeed the Hardy proof of nonlocality has been termed by Mermin as the "simplest form of Bell's theorem" and "stands in its pristine simplicity as one of the strangest and most beautiful gems yet to be found in the extraordinary soil of quantum mechanics" \cite{Mermin94, JXSPC18}. We begin by recalling the analogous Hardy proof for contextuality, first introduced in \cite{CBCB13} with an experimental test performed in \cite{MANCB14}. Consider a physical system of five boxes, numbered from $1$ to $5$, such that each of them can be either empty or full. Let’s denote as $P(0,1|2,3)$ the joint probability of finding box $2$ empty and box $3$ full. One can prepare this system in a state such that $P(0,1|1,2) + P(0,1|2,3) = 1$, and $P(0,1|3,4) + P(0,1|4,5) = 1$. The first condition states that  when box $2$ is full then box $1$ is empty and when box $2$ is empty then box $3$ is full, i.e., $P(1,1|1,2) = P(0,0|2,3) = 0$. Similarly, the second condition is equivalent to $P(1,1|3,4) = P(0,0|4,5) = 0$. From these conditions, and if the result of finding the boxes empty or full were predetermined and independent of which boxes were opened (noncontextuality), one can deduce that $P(0,1|5,1) = 0$. Specifically, considering a noncontextual deterministic outcome probability assignment from $\{0,1\}$ for this experiment (or convex mixtures thereof), we observe that the constraints $P(0,1|1,2) + P(0,1|2,3) = 1$ and $P(0,1|3,4) + P(0,1|4,5) = 1$ imply that $P(0,1|5,1) = 0$. 
On the other hand, suppose one were to prepare a qutrit system in the state $|\psi \rangle = \frac{1}{\sqrt{3}}(1,1,1)^T$, and suppose that opening a box $i = 1,\ldots, 5$ corresponded to measuring the projector on the states $|v_i \rangle$ given as $|v_1 \rangle = \frac{1}{\sqrt{3}}(1,-1,1)^T$, $|v_2 \rangle = \frac{1}{\sqrt{2}}(1,1,0)^T$, $|v_3 \rangle = (0,0,1)^T$, $ |v_4 \rangle = (1,0,0)^T$ and $|v_5 \rangle = \frac{1}{\sqrt{2}}(0,1,1)^T$, 
with empty and full corresponding to obtaining result $0$ and $1$ respectively. Then one can readily verify that the conditions $P(0,1|1,2) + P(0,1|2,3) = 1$, and $P(0,1|3,4) + P(0,1|4,5) = 1$ are both met while $P(0,1|5,1) = \frac{1}{9} > 0$. The above argument thus provides a simple and easily verifiable proof to rule out noncontextual deterministic probability assignments to quantum projections in dimension $d = 3$.

We now proceed to consider alternative \emph{non}deterministic noncontextual probability assignments beyond the $\{0,1\}$ assignments. 
Formally, let $O = \{0, w_1, w_2, \dots, w_q, 1\} \subset [0,1]$ be a finite set of reals such that $0 < w_1 < w_2 < \ldots < w_q < 1$. This set $O$ relaxes the set of classical behaviors and forms an intermediate set between the classical noncontextual (which correspond to $O = \{0,1\}$) and general no-disturbance (which correspond to $O = [0,1]$) models \cite{RSKK12}. An $O$-valuation is a global function $\lambda_O: V \to O$ that assigns a value from $O$ to each projector, satisfying the KS constraints:
\begin{align}
    \sum_{v \in B} \lambda_O(v) \le 1 \;\; \forall B \;\;
    \sum_{v \in B_{\max}} \lambda_O(v) = 1 \;\; \forall B_{\max}.
\end{align}
Here $B$ denotes a measurement basis, and $B_{\max}$ denotes a complete measurement basis consisting of $d$ orthogonal projectors that sum to the identity, i.e., $\sum_{v \in B_{\max}} |v \rangle \langle v| = \mathbb{I}$.
An empirical model $p: V \to [0,1]$ is an $O$-valued empirical model if it can be written as a convex combination of global $O$-valuations. Classical non-contextual models are in this framework termed $\{0,1\}$-valued empirical models. 

Our first result is a family of Hardy tests of contextuality that rules out general $O$-valued empirical models for arbitrary $O$. Specifically, consider a physical system of $6k+2$ boxes labelled $\{v_1^{(m)}, \ldots, v_8^{(m)}\}$ for $m \in \{1,\ldots, k\}$ for a positive integer $k$ and with the identification $v_5^{(m)} \equiv v_1^{(m+1)}$ and $v_8^{(m)} \equiv v_2^{(m+1)}$, where each box can be either empty $(0)$ or full $(1)$. Let $P(0,1|i,j)$ denote the probability of finding box $i$ empty and box $j$ full, and similarly let $P(0,0,1|i,j,l)$ denote the joint probability of finding boxes $i$ and $j$ empty and box $l$ full.
The Hardy constraints are given for every  $m \in \{1, \dots, k\}$ as:
\begin{widetext}
\begin{align}
    &P(0,0,1 \mid v_3^{(m)},v_4^{(m)},v_5^{(m)}) + P(0,1,0 \mid v_3^{(m)},v_4^{(m)},v_5^{(m)}) + P(1,0,0 \mid v_3^{(m)},v_4^{(m)},v_5^{(m)}) = 1, \nonumber \\
    &P(0,0,1 \mid v_6^{(m)},v_7^{(m)},v_8^{(m)}) + P(0,1,0 \mid v_6^{(m)},v_7^{(m)},v_8^{(m)}) + P(1,0,0 \mid v_6^{(m)},v_7^{(m)},v_8^{(m)}) = 1, \nonumber \\
    &P(1,1 \mid v_1^{(m)},v_3^{(m)}) = 0, \quad P(1,1 \mid v_1^{(m)},v_6^{(m)}) = 0, \quad P(1,1 \mid v_2^{(m)},v_4^{(m)}) = 0, \quad P(1,1 \mid v_2^{(m)},v_7^{(m)}) = 0, \nonumber \\
    &v_5^{(m)} \equiv v_1^{(m+1)}, \qquad v_8^{(m)} \equiv v_2^{(m+1)}, \qquad P(1 \mid v_1^{(1)}) = 1.
\end{align}
\end{widetext}
From these conditions, we prove in App. \ref{app:HardyO-context} that anyone who assumes that the result of finding the boxes empty or full is predetermined by a noncontextual $O$-valued empirical model would necessarily conclude that the probability of finding the box $v_2^{(1)}$ full is strictly bounded as  $P_{(O)}(1|v_2^{(1)}) \leq w_q < 1$.
However, one can prepare a qutrit system along with suitable projective measurements (see App. \ref{app:HardyO-context}) such that all the Hardy conditions are strictly satisfied while the quantum value evaluates to $P_{(q)}(1|v_2^{(1)}) = \left(\frac{k}{k+2}\right)^2$. Specifically, suppose that one were to prepare a qutrit system in the state $|v_1^{(1)}\rangle = \left(0, \cos{\theta_k}, \sin{\theta_k}\right)^T$ with $\theta_k = \frac{1}{2} \arcsin{\left(\frac{k}{k + 2}\right)}$. And consider that opening a box corresponded to measuring the projector on the states $|v_i^{(m)} \rangle$ given as: 
\begin{widetext}
\begin{eqnarray}
\label{eq:optq-genm-main}
    &&|v_1^{(m)}\rangle = R^{m-1} \begin{pmatrix} 0 \\ \cos\theta_{k-m+1} \\ \sin\theta_{k-m+1} \end{pmatrix}, \quad  |v_2^{(m)} \rangle = R^{m-1} \begin{pmatrix} 0 \\ \sin{\theta_{k-m+1}} \\ \cos{\theta_{k-m+1}} \end{pmatrix}, \nonumber \\
    &&|v_3^{(m)}\rangle =  \frac{1}{\sqrt{1+k}} R^{m-1}\begin{pmatrix} \sqrt{\frac{k}{2}} (-\cos{\theta_{k-m+1}}+\sin{\theta_{k-m+1}}) \\ -\sin{\theta_{k-m+1}} \\ \cos{\theta_{k-m+1}} \end{pmatrix}, \; |v_4^{(1)}\rangle =  \frac{1}{\sqrt{1+k}} R^{m-1}\begin{pmatrix} \sqrt{\frac{k}{2}} (\cos{\theta_{k-m+1}}-\sin{\theta_{k-m+1}}) \\ -\cos{\theta_{k-m+1}} \\ \sin{\theta_{k-m+1}} \end{pmatrix}, \nonumber \\ 
    &&|v_6^{(m)}\rangle =  \frac{1}{\sqrt{1+k}}R^{m-1} \begin{pmatrix}\sqrt{\frac{k}{2}} (-\cos{\theta_{k-m+1}}+\sin{\theta_{k-m+1}}) \\ \sin{\theta_{k-m+1}} \\ -\cos{\theta_{k-m+1}}\end{pmatrix}, \; |v_7^{(m)}\rangle =  \frac{1}{\sqrt{1+k}}R^{m-1}\begin{pmatrix}\sqrt{\frac{k}{2}} (\cos{\theta_{k-m+1}}-\sin{\theta_{k-m+1}}) \\ \cos{\theta_{k-m+1}} \\ -\sin{\theta_{k-m+1}}\end{pmatrix}, \nonumber \\
    &&|v_5^{(m)}\rangle =  \frac{1}{\sqrt{1+k}}R^{m-1} \begin{pmatrix} 1 \\ \sqrt{\frac{k}{2}} \\ \sqrt{\frac{k}{2}}\end{pmatrix}, \quad |v_8^{(m)}\rangle =\frac{1}{\sqrt{1+k}} R^{m-1} \begin{pmatrix}-1 \\ \sqrt{\frac{k}{2}} \\ \sqrt{\frac{k}{2}}\end{pmatrix}. 
\end{eqnarray}
\end{widetext}
with empty and full corresponding to obtain results $0$ and $1$ respectively. Here
$R$ denotes the rotation matrix
    $R := \begin{pmatrix}
0 & \frac{1}{\sqrt{2}} & -\frac{1}{\sqrt{2}} \\
-\frac{1}{\sqrt{2}} & \frac{1}{2} & \frac{1}{2} \\
\frac{1}{\sqrt{2}} & \frac{1}{2} & \frac{1}{2}
\end{pmatrix}$
and $\theta_j := \frac{1}{2} \arcsin{\left(\frac{j}{j + 2}\right)}$ for $j = 1, \ldots, k$.
Therefore, for any given $O$, selecting a depth $k > \frac{2 \sqrt{w_q}}{1- \sqrt{w_q}}$ implies that $P^{(O)}(1|v_2^{(1)}) <  P^{(q)}(1|v_2^{(1)})$. The above argument thus provides a simple and experimentally feasible Hardy-like test that falsifies the $O$-valued noncontextual ontology.
The minimum number of projective measurements required to rule out an $O$-valued model for $O = \{0,w_1,w_2,\ldots, w_q,1\}$ with $0 < w_1 < w_2 < \ldots < w_q < 1$ depends on the value of $w_q$. The minimum number of projectors $n_{\min}$ is given by:
$$n_{\min} = 6 \times \left\lfloor \frac{2\sqrt{w_q}}{1 - \sqrt{w_q}} \right\rfloor + 8.$$ 
For instance for $w_q = 0.25$ the number of projectors needed is $20$, for $w_q = 0.5$ the number is $32$, for $w_q = 0.75$ the number is $80$ and for $w_q = 0.90$ the number is $224$. 

We remark here that the above test makes the standard minimal assumptions for contextuality tests - that the observables measured obey the specified compatibility structure, that measuring one observable does not physically disturb or alter the system in a way that affects the subsequent outcome of a compatible observable. It is also worth pointing out that existing proofs of contextuality (both state-independent and state-dependent) are insufficient to rule out such $O$-valued models. Already in \cite{Ravi24}, it was pointed out that all existing state-independent proofs admit value assignments from $\{0,1/2,1\}$-models, let alone more general theories. 



While the above result is a state-dependent Hardy test that is experimentally feasible, it is also desirable to achieve state-independent proofs a la Kochen and Specker. In this regard, in Appendix \ref{app:GenKS}, we prove a generalisation of the KS theorem ruling out $O$-valued models for a class of $O$ that while not arbitrary contains previous results as a subclass \cite{Ravi24}. Specifically, we prove the following theorem (for a detailed proof see App. \ref{app:GenKS}).
\begin{theorem}
    Let $O^* = \{0, w_1, w_2, \ldots, w_q, 1\}$ be such that the intersection of $O^3$ with the $2$-simplex is entirely contained within its boundary, i.e., $(O^*)^3 \cap \Delta^2 \subseteq \partial \Delta^2$. There exists a finite set of rays $V_{O^*} \subset \mathbb{C}^3$ that admits no global $O^*$-valued assignment. 
\end{theorem}
Here $\Delta^2 \subset \mathbb{R}^3$ denote the 2-simplex, defined as $\Delta^2 = \left\{ (x, y, z) \in \mathbb{R}^3 \mid x \ge 0, y \ge 0, z \ge 0 \text{ and } x + y + z = 1 \right\}$.  The boundary of $\Delta^2$ is the set of points where at least one coordinate vanishes $\partial \Delta^2 = \left\{ (x, y, z) \in \Delta^2 \mid x=0 \text{ or } y=0 \text{ or } z=0\right\}$. To elaborate, the set of models we consider here are the $O^*$ for which every probability assignment $(p_1,p_2,p_3) \in {O^*}^3$ satisfying the normalization condition $p_1 + p_2 + p_3 = 1$ contains at least one zero. This is a strict generalization of the models considered in \cite{Ravi24} since while $O^* = \{0, p, 1-p, 1\}$ clearly satisfies the condition, there are a large class of other models that also fall under the condition, such as for instance $O^* = \{0, 1/m, 1\}$ for $m \geq 2$. We also remark that the proof of the theorem (in App. \ref{app:GenKS}) is constructive as opposed to the non-constructive proof obtained via the Gleason theorem.

\textit{Applications of the result.-} Thus far, we have seen the abstract generalisation of the KS theorem ruling out a larger class of empirical models than done before. In this section, we delineate the physical significance of the result as establishing a stronger form of contextuality as a resource than hitherto considered, and in terms of ruling out finite many-valued logical models of quantum theory. 

As mentioned earlier, quantum contextuality is a powerful resource for several applications in quantum information, and as such it is important to quantify and understand the strongest manifestations of quantum contextuality. A rigorous hierarchy of contextuality was formulated by Abramsky and Brandenburger in their sheaf-theoretic approach to contextuality in \cite{AB11}. In this framework, contextuality is not seen as simply a violation of a noncontextuality inequality, but more rigorously as a topological obstruction to the existence of a global section (see App. \ref{app:Sheaftheory-O} for the technical details). 

Formally, \cite{AB11} introduced a strict hierarchy in terms of probabilistic vs possibilistic vs strong contextuality. An empirical model $e$ assigning probabilities to measurement outcomes is said to be possibilistically contextual if there exists no global section that is consistent with its local supports. An empirical model $e$ is said to be strongly contextual if no local section in the support of any context can be extended to a consistent global section. In this framework, strong contextuality represents the absolute maximum deviation from classicality - in a strongly contextual model, the fraction of the observed behavior that can be explained by a classical noncontextual ($\{0,1\}$-valued) model is zero (stated in other words, its contextual fraction \cite{ABM17} is one). Abramsky et al. \cite{ABKLM15} also demonstrated that strong contextuality manifests as a non-trivial cohomology class in the first cohomology group providing a purely homological witness for quantum paradoxes. Notably, the quantum behavior achieving the maximum violation of the well-known CHSH test is contextual, but is not strongly contextual. In contrast, the Popescu-Rohrlich (PR) box is seen to be strongly contextual \cite{PR94, AB11} - there are no global sections compatible with its support. 

While the hierarchy is mathematically rigorous, it is fundamentally restricted to a $\{0,1\}$ ontological truth-set. Here we extend this to an even stronger form of contextuality (that we term strong $O$-valued contextuality) where we test empirical models $e$ against the set of $O$-valued models. For an empirical model $e$, we define the $O$-valued contextual fraction $CF_{O}(e)$ as the minimum weight of the empirical model that cannot be decomposed into a convex combination of $O$-valued global sections. We define an empirical model $e$ to be $O$-valued contextual if $CF_{O}(e) > 0$. Furthermore, we define an empirical model $e$ to be Strongly $O$-valued Contextual if no fraction of the empirical behavior can be explained by an $O$-valued global section in the presheaf, i.e., if $CF_{O}(e) = 1$. Let $e_{PR}$ denote the empirical model of the Popescu-Rohrlich box \cite{PR94}. While it exhibits standard Strong Contextuality ($CF_{\{0,1\}}(e_{PR}) = 1)$, we see that it fails to exhibit any $O$-valued contextuality for $O_1 = \{0,1/2,1\}$, that is there exists a single valid global section in $\mathcal{E}_{O_1}(X)$ and $CF_{O_1}(e_{PR}) = 0$. On the other hand, we see that the quantum behaviors that provide optimal violations for the Hardy tests in the previous section exhibit $O$-valued contextuality for any $O = \{0, w_1, \ldots, w_q, 1\}$ such that $\frac{2 \sqrt{w_q}}{1-\sqrt{w_q}} < k$. 
Therefore, the notion of strong $O$-valued contextuality strictly extends the probabilistic-possibilistic-strong contextuality hierarchy of \cite{AB11} and provides a strong resource strengthening the known quantum information applications. While \cite{AB11} proved that contextuality is the obstruction of the Boolean presheaf $\mathcal{E}_{\{0,1\}}$, we prove constructively that quantum theory also evades the global sections for all finite-valued presheaves $\mathcal{E}_{O}$ (for details see App. \ref{app:Sheaftheory-O}).

As a second point of physical interest, we also note an interesting connection between the $O$-valued empirical models and finite many-valued logic. To elaborate, Birkhoff and von Neumann (BvN) in their breakthrough paper \cite{BvN36} argued that the structure of a set of dichotomic (`yes-no') propositions about properties of quantum objects termed `quantum logic', differs from the structure of a set of such propositions pertaining to classical objects in which case it is a Boolean algebra. Since an experimental test of any such dichotomic proposition results in either a true or false outcome, quantum logic in the BvN sense is generally treated as a $2$-valued non-classical logic. At the same time, other non-classical logics were studied, among them being various kinds of many-valued logics \cite{Reichenbach44, Lukasiewicz70, Enderton01, Svozil98, Pavicic92}. In particular, it was argued that statements about future non-certain events belong to the domain of many-valued logic \cite{Lukasiewicz70}. In \cite{Pykacz94}, a faithful isomorphism was established between Birkhoff-von Neumann quantum logic and a specific infinite-valued {\L}ukasiewicz logic. It must be noted that the assignments in many-valued logic are not natively treated as a probability but rather as an objective truth-value assigned to the event (see App. \ref{app:many-valuedlogic-O} for details).

While it is known that finite-valued logics are mathematically untenable for quantum theory as a corollary of Gleason's theorem, that  theorem relies intrinsically on the continuity of the entire sphere, and it is preferable to have an experimentally amenable test that queries only a strictly finite subset of propositions (a finite partial Boolean algebra) to rule out finite-valued logics. The result in the previous section provides a finite partial Boolean algebra that serves to rule out finite-valued logical models as underlying quantum theory. 
        
\textit{Conclusions and Open Questions.-}
Gleason's theorem requires continuity and infinite sets of vectors to pinpoint the Born rule. In contrast, the finite constructions presented here show that one does not need infinity or continuous functions to rule out alternatives to the Born rule. At least some of them (the discrete non-contextual nondeterministic probability assignments) can be ruled out with finite, discrete sets of projectors. 
Richman and Bridges \cite{RB99} proved Gleason's theorem constructively extracting the density matrix $\rho$ from a global probability assignment by a computable procedure. Busch \cite{Busch03} provided a short constructive proof for general POVMs. Our constructions provide a pathway towards finite constructive proofs of Gleason's theorem. Specifically, it provides a tighter bridge for Quantum Logic by showing that the Born Rule can be derived from a finite set of discrete measurements rather than requiring the entire infinite lattice of projections. 

For deterministic $O = \{0,1\}$-valued models, rigorous bounds have been established and records have been found for minimal KS sets in specific Hilbert space dimensions. It is an open question as to what are the minimum cardinality vector sets to rule out general nondeterministic noncontextual models in the state-dependent manner of App. \ref{app:HardyO-context}. While we have derived KS-type proofs for a broad class of $O^*$-valued models, it would be very interesting to complete the research direction by deriving finite KS-type state-independent proofs for arbitrary $O$. It would also be interesting to investigate the scaling behavior of the memory cost of classical simulation \cite{KGP+11} for the constructed quantum contextual sets. Finally, it would be interesting to pursue experimental realisations of (some of) the contextuality tests proposed in this work. In this regard, it is of importance to relax the strict notions of noncontextuality and no-disturbance using techniques from \cite{LR25, KDL15, Fyrillas23} among others. 

\textit{Acknowledgments.-}
We acknowledge support from the General Research Fund (GRF) Grant No.\ 17211122, and the Research Impact Fund (RIF) Grant No.\ R7035-21.

\onecolumngrid
\appendix

\section{Preliminaries: Structure of the noncontextual models }
\label{sec:NCHV}

Multiple ways to define noncontextual hidden-variable models are present in the literature, including the marginal problem definition \cite{Vorobev62, Fine82}, the  graph-theoretic approach of \cite{CSW14}, the sheaf-theoretic approach of \cite{AB11}, the hypergraph-theoretic approach of \cite{AFLS15}, etc. Here we follow the graph-theoretic approach of \cite{CSW14} with its corresponding noncontextuality and no-disturbance polytopes, to explicitly define the $O$-valued hidden variable models that we rule out in this paper. Specifically, we introduce the formal mathematical framework for quantum contextuality scenarios using the language of projective measurements and their combinatorial representation via orthogonality graphs \cite{CSW14}. We note that this formulation restricts to ideal measurements corresponding to rank-one orthogonal projectors acting on a Hilbert space $\mathcal{H}$.

Let $\mathcal{V} = \{v_1, \ldots, v_n\}$ be a finite set of experimental events. In a quantum realization, each event $v_i \in \mathcal{V}$ is assigned a rank-one projector $\Pi_i = \ket{v_i}\bra{v_i}$ acting on a finite-dimensional Hilbert space $\mathcal{H}$. The operational relationships between these events are encoded using an orthogonality graph \cite{LSS87}.

\begin{definition}[Orthogonality Graph]
An orthogonality graph is an undirected graph $G = (V, E)$ where: 
\begin{enumerate}
    \item The vertex set $V$ is identical to the set of events $\mathcal{V}$, where each vertex corresponds to a specific rank-one projector $\Pi_i$.
    \item An edge $e = (v_i, v_j) \in E$ connects two distinct vertices if the corresponding events are mutually exclusive. In quantum theory, this means the projectors are orthogonal: $\Pi_i \Pi_j = 0 \iff \braket{v_i \mid v_j} = 0$.
\end{enumerate}
\end{definition}
Analogously, an orthonormal representation of $G = (V,E)$ in $n$-dimensional Euclidean space is a function $f: V \rightarrow \mathbb{R}^n$ such that $\| f(v) \| = 1$ for all $v \in V$ and $\langle f(u), f(v) \rangle = 0$ for all $\{u,v\} \in E$.
A context $C \subseteq V$ is defined as a maximal clique (a completely connected subgraph) of $G$. Since all vertices within a clique are pairwise connected, their corresponding rank-one projectors are mutually orthogonal and in quantum theory, these can be measured simultaneously as part of a single von Neumann measurement. Let $\mathcal{C}$ denote the set of all maximal cliques of $G$, forming a cover of the vertex set, and denote the set of all maximum cliques of $G$ as $\mathcal{C}_{\max} \subseteq \mathcal{M}$. Recall the graph-theoretic notions: a maximal clique is one that cannot be made bigger by adding another vertex, and a maximum clique is one that has the largest possible number of vertices in the entire graph (known as the clique number $\omega(G)$ of the graph).  

An operational description of the scenario is given by an empirical model, which assigns probabilities to these events.

\begin{definition}[Empirical Model]
An empirical model $e$ on an orthogonality graph $G=(V,E)$ with a clique cover $\mathcal{C}$ is an assignment of probabilities to individual vertices, conditional on the measured context. Formally, it is a set of distributions: $e = \{p_C\}_{C \in \mathcal{C}}$,
where $p_C: C \to [0, 1]$ satisfies the local normalization condition for each maximal clique $C$: $\sum_{v_i \in C} p_C(v_i) \le 1$ and for each maximum clique $C$: $\sum_{v_i \in C} p_C(v_i) = 1$.
The value $p_C(v_i)$ represents the probability that the event $v_i$ occurs given that the context $C$ is chosen.
\end{definition}
Since each context represents a measurement, any valid empirical model must satisfy the constraint ($\le 1$) for all maximal cliques, and a strict normalization condition ($= 1$) for the maximum cliques. 

\subsubsection{General No-disturbance Models}
The physical requirement of no-disturbance \cite{RSKK12} (also called consistency) dictates that the probability of an event $v_i$ happening must be independent of the specific context in which it is measured. Note that in Bell scenarios, the corresponding condition to no-disturbance is the well-known no-signalling principle. 

\begin{definition}[No-disturbance/Consistent models]
An empirical model $e = \{p_C\}_{C \in \mathcal{M}}$ is non-disturbing (or consistent) if, for any two contexts $C, C' \in \mathcal{C}$ sharing a vertex $v_i \in C \cap C'$, the assigned probabilities coincide: $p_C(v_i) = p_{C'}(v_i) =: p(v_i)$.
\end{definition}

In other words, an empirical model satisfying no-disturbance simplifies to a single global assignment $p: V \to [0, 1]$ over the vertices of the graph. The local normalization condition translates to a global constraint: the sum of probabilities over any clique $C$ in the graph cannot exceed unity: $\sum_{v_i \in C} p(v_i) \le 1 \quad \forall C \in \mathcal{C}$, and for maximum cliques $\sum_{v_i \in C} p(v_i) = 1$. 
This system of linear inequalities and equalities defines the no-disturbance/consistent polytope over the orthogonality graph $G$, that we term $\mathcal{ND}(G)$ and define formally using the graph-theoretic clique-constrained stable-set polytope $QSTAB(G)$ \cite{LSS87} as follows.


\begin{definition}[No-disturbance/Consistent Polytope]
The no-disturbance/consistent polytope $\mathcal{P}_{ND}(G)$ is defined by intersecting the clique-constrained stable set polytope $\text{QSTAB}(G)$ \cite{LSS87} with the maximum clique normalization hyperplanes:
\begin{equation}
    \mathcal{P}_{ND}(G) = \text{QSTAB}(G) \cap \left\{ p \in \mathbb{R}^V \;\middle|\; \sum_{v \in C_{\max}} p(v) = 1 \quad \forall C_{\max} \in \mathcal{C}_{\max} \right\},
\end{equation}
where $\text{QSTAB}(G)$ is the clique-constrained polytope defined by $p(v) \ge 0$ for all $v \in V$ and $\sum_{v \in C} p(v) \le 1$ for every maximal clique $C \in \mathcal{C}$.
\end{definition}

\subsubsection{Classical Noncontextual Models}

Classical noncontextual models are formed by taking the convex hull of deterministic assignments that respect both the exclusivity and the normalization constraints in the graph.

\begin{definition}[Deterministic Noncontextual models]
A deterministic noncontextual empirical model is a no-disturbance model where the global function $\lambda: V \to \{0,1\}$ maps each vertex to a value in $\{0,1\}$, and obeys the normalization constraints:
\begin{align}
    \sum_{v \in C} \lambda(v) \le 1 \quad \forall C \in \mathcal{C}, \qquad 
    \sum_{v \in C_{\max}} \lambda(v) = 1 \quad \forall C_{\max} \in \mathcal{C}_{\max}.
\end{align}
We denote the finite set of all such valid binary configurations as $\Lambda_{\text{B}}(G)$.
\end{definition}
The convex hull of all deterministic noncontextual models forms the classical noncontextual polytope $\mathcal{P}_{C}(G)$ that we define in terms of the stable set polytope $STAB(G)$ from graph theory \cite{LSS87} as follows.
\begin{definition}
The classical noncontextual polytope $\mathcal{C}(G)$ is the convex hull of the deterministic noncontextual models:
\begin{equation}
    \mathcal{P}_{C}(G) = \text{conv}(\Lambda_{\text{B}}(G)) = \text{STAB}(G) \cap \left\{ p \in \mathbb{R}^V \;\middle|\; \sum_{v \in C_{\max}} p(v) = 1 \quad \forall C_{\max} \in \mathcal{C}_{\max} \right\},
\end{equation}
where $\text{STAB}(G)$ is the stable set polytope \cite{LSS87}.
\end{definition}


\subsubsection{Noncontextual $O$-Valued Models: An intermediate set}
Quantum correlations lie within the set of no-disturbance behaviors $\mathcal{P}_{ND}(G)$ and separations from the classical set $\mathcal{P}_{C}(G)$ have been well-studied in the literature. In the current paper, we focus on a relaxation of the classical polytope to investigate how far quantum theory deviates from classical. 
To that end, we now relax the binary restriction of classical determinism by allowing the underlying, context-independent ontic states to assign values from a discrete set $O$ defined as follows. 
\begin{definition}
Let $O = \{0, w_1, w_2, \dots, w_q, 1\} \subset [0,1]$ be a finite set of reals such that $0 < w_1 < w_2 < \ldots < w_q < 1$.
\end{definition}
This set $O$ relaxes the set of classical behaviors and forms an intermediate set between the classical noncontextual (which correspond to $O = \{0,1\}$) and general no-disturbance (which correspond to $O = [0,1]$) models. 

\begin{definition}[$O$-Valuation]
An $O$-valuation is a global function $\lambda_O: V \to O$ that assigns a value from $O$ to each vertex, satisfying the KS constraints:
\begin{align}
    \sum_{v \in C} \lambda_O(v) \le 1 \quad \forall C \in \mathcal{C}, \qquad 
    \sum_{v \in C_{\max}} \lambda_O(v) = 1 \quad \forall C_{\max} \in \mathcal{C}_{\max}.
\end{align}
We denote the finite set of all such valid global configurations as $\Lambda_O(G) \subset O^V$.
\end{definition}


\begin{definition}[$O$-Valued Empirical Model]
An empirical model $p: V \to [0,1]$ is an $O$-valued empirical model if it can be written as a convex combination of global $O$-valuations. Formally, there exists a probability distribution $\mu_O: \Lambda_O(G) \to [0,1]$ such that:
\begin{equation}
    p(v_i) = \sum_{\lambda_O \in \Lambda_O(G)} \mu_O(\lambda_O) \cdot \lambda_O(v_i) \quad \forall v_i \in V,
\end{equation}
where $\sum_{\lambda_O \in \Lambda_O(G)} \mu_O(\lambda_O) = 1$. The set of all such models forms a polytope defined by $\mathcal{P}_O(G) = \text{conv}(\Lambda_O(G))$.
\end{definition}
We observe that in this terminology, $\mathcal{P}_{C}(G) \equiv \mathcal{P}_{\{0,1\}}(G)$ and $\mathcal{P}_{ND}(G) \equiv \mathcal{P}_{[0,1]}(G)$.
For instance, selecting $O = \{0, \frac{1}{2}, 1\}$ expands the classical boundary to capture highly contextual, non-signaling behaviors like the well-known Popescu-Rohrlich (PR) box.
s
Since $\{0,1\} \subseteq O \subset [0,1]$, we obtain a hierarchy of polytopes:
\begin{equation}
    \mathcal{P}(G) \subseteq \mathcal{P}_O(G) \subseteq \mathcal{P}_{ND}(G).
\end{equation}

At this point, it is worth remarking that quantum theory is also noncontextual, in that the probability assignment for a measurement of a projector $P$ on a state $\rho$ is given as $\text{Tr}[\rho P]$ independent of the specific basis in which the projector is measured. As such the set of quantum correlations $\mathcal{P}_{Q}(G)$ for any given orthogonality graph is contained within the no-disturbance polytope $\mathcal{P}_{Q}(G) \subset \mathcal{P}_{ND}(G)$. 

The question we are interested in is to engineer orthogonality graphs $G$ and suitable quantum correlations that lie outside $\mathcal{P}_O(G)$ for any given $O$. As such, it is worth remarking that for most of the studied orthogonality graphs in the literature the set of quantum correlations  $\mathcal{P}_{Q}(G) \subset \mathcal{P}_O(G)$ already for $O = \{0, 1/2, 1\}$. For instance, for the Hardy test from \cite{CBCB13} described in the main text, the no-disturbance polytope has been characterised. It is known that apart from the classical deterministic vertices, the only non-classical vertex is the one which assigns probability $1/2$ to every vertex of the pentagon $C_5$ \cite{Wright78}. In particular, Wright proved that one can mathematically construct a state on this pentagon logic where every single vertex has a "Yes" probability of exactly \(\frac{1}{2}\). In other words, while the quantum value for the Hardy test for $C_5$ was found to be $P_{(q)}(0,1|5,1) = 1/9 > 0$ in \cite{CBCB13}, the no-disturbance value for the test is $1/2 > 1/9$, showing that the test cannot rule out a $\{0,1/2,1\}$-empirical model. This observation can also be extended to general $n$-cycle scenarios where again it is known that $\mathcal{P}_{ND}(G) = \mathcal{P}_{O}(G)$ for $O = \{0, 1/2, 1\}$ \cite{AQBCC13}. Similarly known KS proofs such as the Peres-Mermin magic square and the Mermin pentagram as well as the original KS proof \cite{KS67} were found to admit value assignments from $O = \{0, 1/2, 1\}$. As such, the question of ruling out general $O$-valued empirical models has remained open (beyond the specific $O$ considered in \cite{Ravi24}).



\section{A Hardy-like test of contextuality to rule out arbitrary $O$-valued empirical models}
\label{app:HardyO-context}
In this section, we formulate a Hardy-type test of contextuality to rule out general $O$-valued empirical models. The test is state-dependent and experimentally feasible, and makes use of a specific orthogonality graph that we term the nested Clifton graph. 
\subsection{The Nested Clifton Graph and its optimal orthogonal representation}
Consider the orthogonality graph that results from a nesting of the well-known Clifton graph \cite{Clifton93} used by Kochen and Specker in their original proof \cite{KS67}. This Nested Clifton graph is shown in Fig. \ref{fig:clifton_nested} and its structure is as follows. 

\begin{figure}[htbp]
    \centering
    \includegraphics[width=0.6\textwidth]{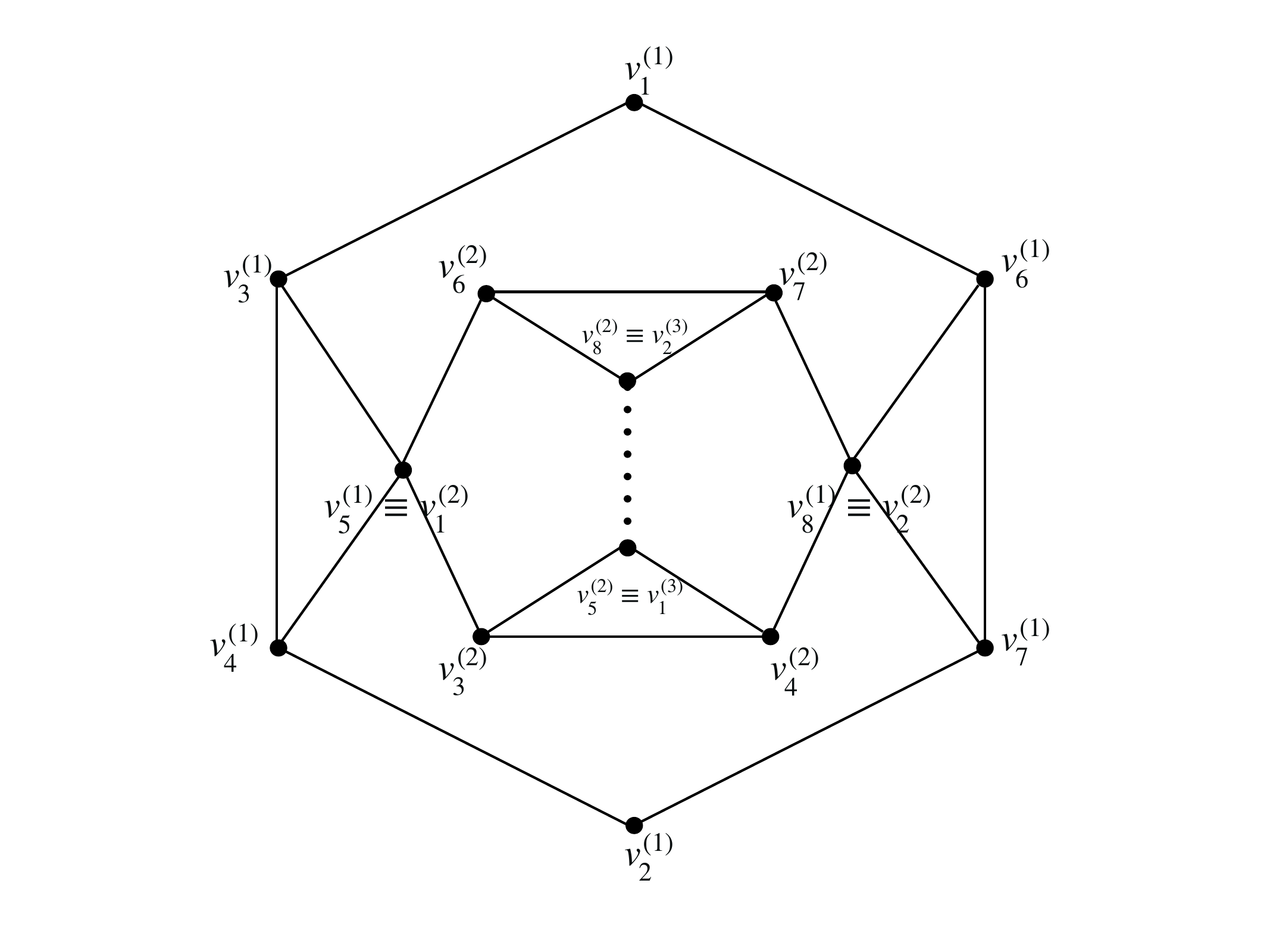}
    \caption{The $k$-nested Clifton graph with the recursive structural nesting as explained in the text. The dotted central line indicates that the recursive nesting with the $m+1$-th nesting replacing the edge $\{v_5^{(m)}, v_8^{(m)}\}$ with a Clifton graph as explained in the text. The optimal quantum (and no-disturbance) values for the Hardy test of contextuality built using this graph are derived in Thms. \ref{thm:optq-nestClif} and \ref{thm:NestClif-NDval}.}
    \label{fig:clifton_nested}
\end{figure}

Let $k \in \mathbb{Z}_+$ denote the total depth of nesting. The outer layer $m=1$ contains the vertices $\{v_1^{(1)}, v_2^{(1)}, \ldots, v_8^{(1)}\}$. Each layer is isomorphic to the Clifton graph, so that for each $m = 2, \ldots, k$ we introduce vertices $\{v_1^{(m)}, v_2^{(m)}, \ldots, v_8^{(m)}\}$ with the following identification for $m > 1$:
\begin{equation}
    v_1^{(m)} = v_5^{(m-1)} \; \; \;\text{and} \; \; \; v_2^{(m)} = v_8^{(m-1)}.
\end{equation}
That is, the central vertices $v_5$ and $v_8$ of the $(m-1)$-th layer form the two outer vertices $v_1$ and $v_2$ of the $m$-th layer. Let $G_k = (V_k, E_k)$ denote the $k$-nested Clifton graph, constructed by $k$ recursive embeddings of the Clifton gadget. 
The vertex set $V_k$ consists of $|V_k| = 6k+2$ vertices with the original eight-vertex Clifton bug corresponding to the case $k=1$. The edge set $E_k$ consists of $|E_k| = 10k+1$ edges with eleven edges in the Clifton bug corresponding to $k=1$.                                                           

We first observe that for the $k$-nested Clifton graph $G_k$ for any $k \in \mathbb{Z}_+$, it holds in any consistent/no-disturbance theory that $f(v_1^{(1)}) = 1 \implies f(v_2^{(1)}) < 1$. Specifically, we derive the following proposition by writing the no-disturbance value as a linear program. 
\begin{theorem}
\label{thm:NestClif-NDval}
 Let $G_k = (V_k, E_k)$ denote the $k$-nested Clifton graph of Fig. \ref{fig:clifton_nested}. For any consistent probability assignment $f: V_k \rightarrow [0,1]$ satisfying $f(v_1^{(1)})=1$ it holds that 
\begin{equation}
    f(v_2^{(1)}) \leq 1 - 2^{-k}.
\end{equation}
\end{theorem}
\begin{proof}
The proof follows straightforwardly by induction. Consider the base case $k=1$. We are interested in the value of the following linear program for any assignment $f: V_1 \rightarrow [0,1]$: maximize $\; f(v_1^{(1)}) + f(v_2^{(1)})$ subject to the constraints: $f(v_1^{(1)}) + f(v_3^{(1)}) \leq 1$, $f(v_1^{(1)}) + f(v_6^{(1)}) \leq 1$, $f(v_2^{(1)}) + f(v_4^{(1)}) \leq 1$, $f(v_2^{(1)}) + f(v_7^{(1)}) \leq 1$, $f(v_5^{(1)}) + f(v_8^{(1)}) \leq 1$, $f(v_3^{(1)}) + f(v_4^{(1)}) + f(v_5^{(1)})=1$ and $f(v_6^{(1)}) + f(v_7^{(1)}) + f(v_8^{(1)}) = 1$. 
 The constraints come from the orthogonality structure of the graph $G_1$ and the two maximum cliques (of size $\omega(G_1)=3$) for which the normalization condition applies. By adding the inequality constraints and substituting the normalization conditions, one can readily deduce that $2f(v_1^{(1)}) + 2f(v_2^{(1)}) \leq 3$. So that when $f(v_1^{(1)}) = 1$ it holds that $f(v_2^{(1)}) \leq \frac{1}{2}$. Note that this is the maximum consistent (no-disturbance) value of the probability assigned to $v_2^{(1)}$ when $v_1^{(1)}$ has probability one. 

 Now suppose that the statement holds for $k \in \mathbb{Z}_+$, i.e., that $f(v_1^{(1)}) + f(v_2^{(1)}) \leq 2 - \frac{1}{2^k}$ for any assignment $f: V_k \rightarrow [0,1]$. We'll see that the statement holds for $k+1$. In this case we are interested in the value of the following linear program: 
 \begin{eqnarray}
    \max \; &&f(v_1^{(1)}) + f(v_2^{(1)}) \nonumber \\
    \text{s.t.} \; &&f(v_1^{(1)}) + f(v_3^{(1)}) \leq 1, \quad f(v_1^{(1)}) + f(v_6^{(1)}) \leq 1, \nonumber \\
    &&f(v_2^{(1)}) + f(v_4^{(1)}) \leq 1, \quad f(v_2^{(1)}) + f(v_7^{(1)}) \leq 1, \quad f(v_5^{(1)}) + f(v_8^{(1)}) \leq 2 - \frac{1}{2^k}, \nonumber \\
    && f(v_3^{(1)}) + f(v_4^{(1)}) + f(v_5^{(1)})=1, \quad f(v_6^{(1)}) + f(v_7^{(1)}) + f(v_8^{(1)}) = 1.
\end{eqnarray}
Here, the constraint $f(v_5^{(1)}) + f(v_8^{(1)}) \leq 1 - \frac{1}{2^k}$ comes from the inductive hypothesis that for the graph $G_k$ the sum of the assignments to the two outer vertices is bounded by $2 - \frac{1}{2^k}$, and that $v_5^{(1)} = v_1^{(2)}$ and $v_8^{(1)} = v_2^{(2)}$. Now, again adding the inequality constraints and utilizing the normalization constraints gives that
\begin{equation}
    2f(v_1^{(1)}) + 2f(v_2^{(1)}) \leq 4 - \frac{1}{2^k}.
\end{equation}
In other words, when $f(v_1^{(1)})=1$ we have that $f(v_2^{(1)}) \leq 1 - \frac{1}{2^{k+1}}$. We thus see that if the proposition holds for $k$ it holds for $k+1$, so that by induction, the statement holds for any $k \in \mathbb{Z}_+$.   
\end{proof}
\begin{rem}
\label{rem:Omax}
Thm. \ref{thm:NestClif-NDval} shows that for any probability assignment $f: V_k \rightarrow [0,1]$ satisfying $f(v_1^{(1)}) = 1$ it holds that $f(v_2^{(1)}) \leq 1 - \frac{1}{2^k}$. Consider the case of $O$-valued assignments $f: V_k \rightarrow O$ where $O = \{0,w_1,w_2, \ldots, w_q, 1\}$ with $0 < w_1 < w_2 < \ldots < w_q < 1$. Let $\tilde{q}$ denote the largest index $i$ for which $w_i \leq 1 - \frac{1}{2^k}$. Then, for any $O$-valued assignment  
satisfying $f(v_1^{(1)})=1$ it holds that $f(v_2^{(1)}) \leq w_{\tilde{q}}$. 
\end{rem}

Now, let us proceed to identify the corresponding ``Hardy-type'' probability in the case of quantum theory. For this, we need to solve the rank-constrained optimization problem giving the optimal orthogonal representation of the Nested Clifton graph $G_k$. An ansatz solution was given for this problem in \cite{RRHS+20}, here we provide a rigorous proof showing the optimality of this solution. This will enable us identify the minimal $k$ for which the graph $G_k$ serves to rule out $O$-valued assignments.   


\begin{theorem}
\label{thm:optq-nestClif}
 Let $G_k = (V_k, E_k)$ denote the $k$-nested Clifton graph of Fig. \ref{fig:clifton_nested}, and let $v_1^{(1)}, v_2^{(1)} \in V_k$ denote the two outermost end vertices. For any orthonormal representation $f: V_k \rightarrow \mathbb{R}^3$ of $G_k$, it holds that 
 \begin{eqnarray}
     |\langle f(v_1^{(1)} | f(v_2^{(1)} \rangle| \leq \frac{k}{k+2}.
 \end{eqnarray}
\end{theorem}

\begin{proof}
To find the optimal representation analytically, we formulate the geometric constraints of the graph as a rank-constrained positive semidefinite optimization problem. To that end, let us define an extended index set that includes a handle for the preparation state:
\begin{eqnarray}
    \mathcal{I} = \{0\} \cup \big\{ (i,m) \; | i \in \{1,\ldots, 8\}, m \in \{1,\ldots, k\} \big\}, 
\end{eqnarray}
where the index $0$ corresponds to the state $|\psi \rangle$ and each $(i,m)$ corresponds to a vector $|v_i^{(m)} \rangle$. 
We define the Gram matrix $X \in \mathbb{R}^{|\mathcal{I}| \times |\mathcal{I}|}$ of the vector set such that $M_{(i,m),(j,m')} = \langle v_i^{(m)} | v_j^{(m')} \rangle$ for all $(i,m), (j,m') \in \mathcal{I}$. Finding the maximum overlap $\langle \psi | v_2^{(1)} \rangle$ is equivalent to maximizing the matrix entry $M_{0,(2,1)}$. The problem is thus equivalent to the following non-convex rank-constrained semidefinite program:
\begin{eqnarray}
\label{eq:NestedClif-primal}
\text{Maximize}: && M_{0,(2,1)} \nonumber \\
\text{subject to}: &&M \succeq 0, \quad \text{rank}(M) \leq 3, \nonumber \\
&&M_{x,x} = 1 \; \; \forall x \in \mathcal{I}, \quad M_{0,(1,1)} = 1, \nonumber \\
&&M_{(1,m),(5,m-1)} = 1 \; \; \forall m \in \{2,\ldots, k\}, \quad M_{(2,m),(8,m-1)} = 1 \; \; \forall m \in \{2,\ldots, k\}, \nonumber \\
&&M_{(i,m),(j,m)} = 0 \; \; \forall m \in \{1,\ldots, k\} \; \forall \{i,j\} \in E_C, \nonumber \\
&&M_{(5,k),(8,k)} = 0. 
\end{eqnarray}
Since $M$ is positive semidefinite and has ones on the diagonal, the constraint $M_{(1,m),(5,m-1)} = 1$ strictly implies $|v_1^{(m)} \rangle = |v_5^{(m-1)} \rangle$ and the constraint $M_{(2,m),(8,m-1)} = 1$ strictly implies $|v_2^{(m)} \rangle = |v_8^{(m-1)} \rangle$. Similarly the constraint $M_{0,(1,1)} = 1$ enforces that $|\psi \rangle = |v_1^{(1)} \rangle$. $E_C$ denotes the edge set of the Clifton graph, with the constraint $M_{(i,m),(j,m)} = 0$ for all $\{i,j\} \in E_C$ enforcing the orthogonality constraint for each layer $m \in [k]$. The constraint $\text{rank}(M) \leq 3$ restricts the embedding to $\mathbb{R}^3$.

We now proceed to solve the rank-constrained semidefinite program in \eqref{eq:NestedClif-primal}. While duality is not helpful here due to the rank-constraint (see the remark following the proof), we proceed to prove that the optimal solution to \eqref{eq:NestedClif-primal} is $\frac{k}{k+2}$ by induction. The base case $k=1$ is well-known - the maximal overlap of the end vectors in the Clifton bug is known to be $1/3$. We establish the following lemma.
\begin{lemma}
\label{lem:overlap}
    For any valid orthogonal representation $f: V_k \to \mathbb{R}^3$ of $G_k$, it holds that $x_k = \big|\langle v_1^{(k)}|v_2^{(k)}\rangle\big|$ is a strictly monotonically increasing function of $x_{k+1} = \big|\langle v_1^{(k+1)}|v_2^{(k+1)}\rangle\big|= \big|\langle v_5^{(k)}|v_8^{(k)}\rangle\big|$.
\end{lemma}
\begin{proof}
    We have $x_k = \big|\langle v_1^{(k)}|v_2^{(k)}\rangle\big|$ and $x_{k+1} = \big|\langle v_1^{(k+1)}|v_2^{(k+1)}\rangle\big| = \big|\langle v_5^{(k)}|v_8^{(k)}\rangle\big|$. From the orthogonality constraints $\langle v_3^{(k)} | v_1^{(k)} \rangle = 0$, $\langle v_4^{(k)} | v_2^{(k)} \rangle = 0$ and $\langle v_4^{(k)} | v_3^{(k)} \rangle = 0$, we obtain the parametrization:
    \begin{eqnarray}
        |v_3^{(k)}\rangle = \frac{1}{\sqrt{1+x_k}} \left(s,-c,\sqrt{x_k}\right)^T, \quad |v_4^{(k)}\rangle = \frac{1}{\sqrt{1+x_k}} \left(s,c,-\sqrt{x_k}\right)^T,
    \end{eqnarray}
where $c^2 = (1+x_k)/2$ and $s^2 = (1-x_k)/2$. Since $\{|v_3^{(k)}\rangle, |v_4^{(k)}\rangle, |v_5^{(k)}\rangle\}$ forms an orthonormal basis, we have that $|v_5^{(k)}\rangle \propto |v_3^{(k)}\rangle \times |v_5^{(k)}\rangle$, and similarly for $|v_8^{(k)}\rangle \propto |v_6^{(k)}\rangle \times |v_7^{(k)}\rangle$. We obtain the parametrization:
\begin{eqnarray}
        |v_5^{(k)}\rangle = \frac{1}{\sqrt{1+x_k}} \left(0,\sqrt{2 x_k},\sqrt{1-x_k}\right)^T, \quad |v_8^{(k)}\rangle = \frac{1}{\sqrt{1+x_k}} \left(0,-\sqrt{2 x_k},\sqrt{1-x_k}\right)^T.
    \end{eqnarray}
We evaluate the inner product as $x_{k+1} = \big|\langle v_5^{(k)}|v_8^{(k)}\rangle\big| = \frac{3x_k-1}{x_k+1}$ giving $x_k = \frac{x_{k+1}+1}{3-x_{k+1}}$. We see that $\frac{d x_{k}}{d x_{k+1}} = \frac{4}{(3-x_{k+1})^2} > 0$, showing that $x_k$ is a strictly monotonically increasing function of $x_{k+1}$. 
\end{proof}
As the inductive hypothesis, suppose that the overlap at the $k$-th level of nesting equals $k/(k+2)$. We use the Lemma \ref{lem:overlap} to prove that the overlap at the $k+1$-th level equals $(k+1)/(k+3)$. This is seen by the proof of the lemma, where we have $\big|\langle v_1^{(k)}|v_2^{(k)}\rangle\big| = \frac{\big|\langle v_5^{(k)}|v_8^{(k)}\rangle\big|+1}{3-\big|\langle v_5^{(k)}|v_8^{(k)}\rangle\big|} = \frac{\frac{k}{k+2}+1}{3-\frac{k}{k+2}} = \frac{k+1}{k+3}$ establishing the theorem by induction. 
An orthonormal representation in $\mathbb{R}^3$ achieving the optimal overlap exists for the $k$-nested Clifton graph $G_k$, and it has a very elegant structure as we shall show below. 

Let $R$ denote a $\pi/2$-counterclockwise rotation about the axis $|u\rangle = \frac{1}{\sqrt{2}}\left(0,1,1\right)^T$ explicitly given as
\begin{eqnarray}
    R := \begin{pmatrix}
0 & \frac{1}{\sqrt{2}} & -\frac{1}{\sqrt{2}} \\
-\frac{1}{\sqrt{2}} & \frac{1}{2} & \frac{1}{2} \\
\frac{1}{\sqrt{2}} & \frac{1}{2} & \frac{1}{2}
\end{pmatrix}
\end{eqnarray}
and let $\theta_j = \frac{1}{2} \arcsin{\left(\frac{j}{j + 2}\right)}$ for $j = 1, \ldots, k$. The optimal orthogonal representation is given as follows: the vectors in the outer layer $m=1$ are
\begin{eqnarray}
\label{eq:optq-meq1}
    &&|v_1^{(1)}\rangle = \left(0, \cos{\theta_k}, \sin{\theta_k}\right)^T, \quad |v_2^{(1)} \rangle = \left(0, \sin{\theta_k}, \cos{\theta_k}\right)^T, \nonumber \\
    &&|v_3^{(1)}\rangle = \frac{1}{\sqrt{1+k}}\left(\sqrt{\frac{k}{2}} (-\cos{\theta_k}+\sin{\theta_k}), -\sin{\theta_k}, \cos{\theta_k}\right)^T, \; |v_4^{(1)}\rangle = \frac{1}{\sqrt{1+k}}\left(\sqrt{\frac{k}{2}} (\cos{\theta_k}-\sin{\theta_k}), -\cos{\theta_k}, \sin{\theta_k}\right)^T, \nonumber \\ 
    &&|v_6^{(1)}\rangle = \frac{1}{\sqrt{1+k}}\left(\sqrt{\frac{k}{2}} (-\cos{\theta_k}+\sin{\theta_k}), \sin{\theta_k}, -\cos{\theta_k}\right)^T, \; |v_7^{(1)}\rangle = \frac{1}{\sqrt{1+k}}\left(\sqrt{\frac{k}{2}} (\cos{\theta_k}-\sin{\theta_k}), \cos{\theta_k}, -\sin{\theta_k}\right)^T, \nonumber \\
    &&|v_5^{(1)}\rangle = \frac{1}{\sqrt{1+k}}\left(1, \sqrt{\frac{k}{2}}, \sqrt{\frac{k}{2}}\right)^T, \quad |v_8^{(1)}\rangle = \frac{1}{\sqrt{1+k}}\left(-1, \sqrt{\frac{k}{2}}, \sqrt{\frac{k}{2}}\right)^T. 
\end{eqnarray}
We immediately observe that $\big|\langle v_1^{(1)}| v_2^{(1)}\rangle \big| = \sin{(2\theta_k)} = \frac{k}{k+2}$, and $\big|\langle v_1^{(2)}| v_2^{(2)}\rangle \big| = \big|\langle v_5^{(1)}| v_8^{(1)}\rangle \big| = \frac{k-1}{k+1}$. 
The set of vectors in every successive layer $m$ for $m = 2,\ldots, k$ is obtained as follows:
\begin{eqnarray}
\label{eq:optq-genm}
    &&|v_1^{(m)}\rangle = R^{m-1} \begin{pmatrix} 0 \\ \cos\theta_{k-m+1} \\ \sin\theta_{k-m+1} \end{pmatrix}, \quad  |v_2^{(m)} \rangle = R^{m-1} \begin{pmatrix} 0 \\ \sin{\theta_{k-m+1}} \\ \cos{\theta_{k-m+1}} \end{pmatrix}, \nonumber \\
    &&|v_3^{(m)}\rangle =  \frac{1}{\sqrt{1+k}} R^{m-1}\begin{pmatrix} \sqrt{\frac{k}{2}} (-\cos{\theta_{k-m+1}}+\sin{\theta_{k-m+1}}) \\ -\sin{\theta_{k-m+1}} \\ \cos{\theta_{k-m+1}} \end{pmatrix}, \; |v_4^{(1)}\rangle =  \frac{1}{\sqrt{1+k}} R^{m-1}\begin{pmatrix} \sqrt{\frac{k}{2}} (\cos{\theta_{k-m+1}}-\sin{\theta_{k-m+1}}) \\ -\cos{\theta_{k-m+1}} \\ \sin{\theta_{k-m+1}} \end{pmatrix}, \nonumber \\ 
    &&|v_6^{(m)}\rangle =  \frac{1}{\sqrt{1+k}}R^{m-1} \begin{pmatrix}\sqrt{\frac{k}{2}} (-\cos{\theta_{k-m+1}}+\sin{\theta_{k-m+1}}) \\ \sin{\theta_{k-m+1}} \\ -\cos{\theta_{k-m+1}}\end{pmatrix}, \; |v_7^{(m)}\rangle =  \frac{1}{\sqrt{1+k}}R^{m-1}\begin{pmatrix}\sqrt{\frac{k}{2}} (\cos{\theta_{k-m+1}}-\sin{\theta_{k-m+1}}) \\ \cos{\theta_{k-m+1}} \\ -\sin{\theta_{k-m+1}}\end{pmatrix}, \nonumber \\
    &&|v_5^{(m)}\rangle =  \frac{1}{\sqrt{1+k}}R^{m-1} \begin{pmatrix} 1 \\ \sqrt{\frac{k}{2}} \\ \sqrt{\frac{k}{2}}\end{pmatrix}, \quad |v_8^{(m)}\rangle =\frac{1}{\sqrt{1+k}} R^{m-1} \begin{pmatrix}-1 \\ \sqrt{\frac{k}{2}} \\ \sqrt{\frac{k}{2}}\end{pmatrix}. 
\end{eqnarray}
Here $R^{m-1}$ is essentially $R^{(m-1) \; \text{mod} \; 4}$ since $R^4 = \mathbb{I}$. Note though that the vector to which the rotation is applied differs at each layer due to the dependence on $\theta_{j}$. By explicit calculation we see that $|v_1^{(2)}\rangle = |v_5^{(1)}\rangle$ and $|v_2^{(2)}\rangle = |v_8^{(1)}\rangle$.
We have thus derived the optimal orthonormal representation for the family of nested Clifton graphs with the maximum overlap between the outer vertices for any level of nesting. 
\end{proof}

We remark that the optimization problem in \eqref{eq:NestedClif-primal} being non-convex, the problem is not easily solved using duality as is standard for semidefinite programming problems. Specifically, as we shall show by writing the dual problem, there exists a duality gap so that the optimization problem is not exactly solvable by just finding a feasible solution to the dual.

Let $\mathcal{S}^N$ denote the set of real, symmetric $N \times N$ matrices, with the inner product $\langle A, B \rangle = \operatorname{Tr}(AB)$ and let $\mathcal{D}$ define the non-convex primal feasible domain for the optimization problem in \eqref{eq:NestedClif-primal}. We define a linear operator $\mathcal{A} : \mathcal{S}^N \to \mathbb{R}^K$ that encapsulates all $K$ equality constraints in the optimization. The primal problem takes the standard form: $\text{Maximize}: \langle C, M \rangle$ subject to $\mathcal{A}(M)=b, \text{and} \; M \in \mathcal{D}$. Here, the objective matrix $C \in \mathcal{S}^N$ is defined with entries $C_{0, (2,1)} = C_{(2,1), 0} = \frac{1}{2}$ and zeros elsewhere, so that $\langle C, M \rangle = M_{0, (2,1)}$. The vector $b \in \mathbb{R}^K$ contains the right-hand side values for each constraint ($1$ for normalization, $0$ for orthogonality).

We introduce the vector of dual variables $z \in \mathbb{R}^K$ corresponding to the equality constraints. The Lagrangian $L : \mathcal{S}^N \times \mathbb{R}^K \to \mathbb{R}$ is given as $L(M,z) = \langle C, M \rangle + \langle z, b- \mathcal{A}(M) \rangle = \langle b, z \rangle + \langle C- \mathcal{A}^*(z), M \rangle$, where $\mathcal{A}^*$ denotes the adjoint of $\mathcal{A}$. The Lagrangian dual function $g(z)$ is the supremum of $L(M, z)$ over the primal domain $\mathcal{D}$: $g(Z) = \sup_{M \in \mathcal{D}} L(M, z) = \langle b, z \rangle + \sup_{M \in \mathcal{D}} \langle C - \mathcal{A}^*(z), M \rangle$. Defining the dual slack matrix $S \in \mathcal{S}^N$ as $S = \mathcal{A}^*(z) - C$, we must evaluate $\sup_{M \in D} \langle - S, M \rangle$. 
\begin{lemma}
    The quantity $\sup_{M \in \mathcal{D}} \langle -S, M \rangle$ is finite (and exactly $0$) if and only if $S \succeq 0$.
\end{lemma}
\begin{proof}
Suppose that $S$ is not positive semidefinite, then $S$ has at least one negative eigenvalue $\lambda < 0$ with corresponding normalized eigenvector $v \in \mathbb{R}^N$. Defining a test matrix $\tilde{M} = t v v^T$ for some $t > 0$, we see that $\tilde{M} \succeq 0$ and $\text{rank}(\tilde{M}) = 1 < 3$ so that $\tilde{M} \in \mathcal{D}$. Computing $\langle -S, \tilde{M} \rangle = - t \lambda$ we see that as $t \to \infty$ we have that $-t \lambda \to \infty$. Therefore the dual function is unbounded if $S \not \succeq 0$. When $S \succeq 0$ we have that $\langle - S, M \rangle \leq 0$ since $M \succeq 0$. Since $M' = \mathbf{0} \in \mathcal{D}$, the supremum is strictly achieved at $0$.
\end{proof}
Thus, the dual function evaluates to: 
\begin{eqnarray}
    g(z) = \begin{cases}
    \langle b, z \rangle & \text{if } \mathcal{A}^*(z) - C \succeq 0 \\
    \infty   & \text{otherwise}
\end{cases}
\end{eqnarray}
Let us write the dual vector $z \in \mathbb{R}^K$ explicitly as the vector of Langrage multipliers $\{\{y_x\}_{x \in \mathcal{I}}, \lambda, \{\mu_m\}_{m=2}^k, \{\nu_m\}_{m=2}^{k}, \{w_{u,v}\}_{\{u,v\} \in E_k} \}$. Here $\lambda$ is the Langrage multiplier corresponding to the constraint $M_{0,(1,1)} = 1$, $y_x$ is the multiplier corresponding to the constraint $M_{x,x} = 1$, $\mu_m$ and $\nu_m$ are the multipliers corresponding to the constraints $M_{(1,m),(5,m-1)} =1$ and $M_{(2,m),(8,m-1)} = 1$, and $w_{u,v}$ is the multiplier corresponding to the orthogonality constraint $M_{u,v} = 0$ for edge $\{u,v\} \in E_k$. Let $E_{i,j}$ be the indicator matrix with a $1$ at position $(i,j)$ and $0$ elsewhere, we have that 
\begin{eqnarray}
    \mathcal{A}^*(z) = &&\sum_{x \in \mathcal{I}} y_x E_{x,x} + \lambda \left(E_{0,(1,1)} + E_{(1,1),0} \right) + \sum_{m=2}^k \mu_m \left(E_{(1,m),(5,m-1)} + E_{(5,m-1),(1,m)} \right) \nonumber \\
    &&+ \sum_{m=2}^k \nu_m \left(E_{(2,m),(8,m-1)} + E_{(8,m-1),(2,m)} \right) + \sum_{\{u,v\} \in E_k} w_{u,v} \left(E_{u,v} + E_{v,u} \right). 
\end{eqnarray}
Finally, we write the dual optimization problem as:
\begin{eqnarray}
    \text{Minimize}: &&d = 2 \lambda + \sum_{x \in \mathcal{I}} y_x + 2 \sum_{m=2}^k (\mu_m + \nu_m) \nonumber \\
    \text{subject to:} && S \succeq 0, \nonumber \\
    &&S_{x,x} = y_x \; \; \forall x \in \mathcal{I}, \nonumber \\
    &&S_{0,(1,1)}=S_{(1,1),0} = \lambda, \nonumber \\
    &&S_{(1,m),(5,m-1)} = S_{(5,m-1),(1,m)} = \mu_m \; \; \forall m \in \{2, \ldots, k\}, \nonumber \\
     &&S_{(2,m),(8,m-1)} = S_{(8,m-1),(2,m)} = \nu_m \; \; \forall m \in \{2, \ldots, k\}, \nonumber \\
     &&S_{v,u} = S_{u,v} = w_{u,v} \; \; \forall \{u,v\} \in E_k \nonumber \\
     &&S_{0,(2,1)} = S_{(2,1),0} = - \frac{1}{2}, \; \; \; S_{i,j} = 0 \; \; \text{for all other} \; i,j. 
\end{eqnarray}
Let $p^*$  be the optimal value of the rank-$3$ primal problem and $d^*$ be the optimal value of the dual problem derived above. By weak duality we have $p^* \leq d^*$, and the quantity $\Delta = d^* - p^*$ is the duality gap. For this problem, we find that $d^*$ is strictly greater than $p^* = \frac{k}{k+2}$ so that $\Delta > 0$. Therefore, the exact solution to the optimization cannot be obtained by finding a dual feasible solution alone.

\subsection{A Hardy-type test to rule out nondeterministic noncontextual models}
We now translate the above Thm. \ref{thm:optq-nestClif} and Corollary \ref{rem:Omax} into an experimentally feasible test to rule out noncontextual $O$-valued empirical models for any finite $O$.  Specifically, we formulate the above theorems into a Hardy-like proof of $O$-valued contextuality, akin to the Hardy test of (usual $\{0,1\}$-valued) contextuality in \cite{CBCB13, MANCB14}.

Consider a physical system of $6k+2$ boxes, where each box can be either empty $(0)$ or full $(1)$. Let $P(0,1|i,j)$ denote the probability of finding box $i$ empty and box $j$ full, and similarly let $P(0,0,1|i,j,l)$ denote the joint probability of finding boxes $i$ and $j$ empty and box $l$ full.
The Hardy constraints are given for every  $m \in \{1, \dots, k\}$ as:
\begin{align}
    &P(0,0,1 \mid v_3^{(m)},v_4^{(m)},v_5^{(m)}) + P(0,1,0 \mid v_3^{(m)},v_4^{(m)},v_5^{(m)}) + P(1,0,0 \mid v_3^{(m)},v_4^{(m)},v_5^{(m)}) = 1, \nonumber \\
    &P(0,0,1 \mid v_6^{(m)},v_7^{(m)},v_8^{(m)}) + P(0,1,0 \mid v_6^{(m)},v_7^{(m)},v_8^{(m)}) + P(1,0,0 \mid v_6^{(m)},v_7^{(m)},v_8^{(m)}) = 1, \nonumber \\
    &P(1,1 \mid v_1^{(m)},v_3^{(m)}) = 0, \quad P(1,1 \mid v_1^{(m)},v_6^{(m)}) = 0, \quad P(1,1 \mid v_2^{(m)},v_4^{(m)}) = 0, \quad P(1,1 \mid v_2^{(m)},v_7^{(m)}) = 0, \nonumber \\
    &v_5^{(m)} \equiv v_1^{(m+1)}, \qquad v_8^{(m)} \equiv v_2^{(m+1)}, \qquad P(1 \mid v_1^{(1)}) = 1.
\end{align}
From these conditions, anyone who assumes that the result of finding the boxes empty or full is predetermined by a finite noncontextual $O$-valued empirical model - where $O = \{0,w_1,\ldots, w_q,1\}$ with $0 < w_1 < \ldots < w_q < 1$ - would necessarily conclude that the probability of finding the box $v_2^{(1)}$ full is strictly bounded as (following Thm. \ref{thm:NestClif-NDval} and Remark \ref{rem:Omax})
\begin{equation}
    P^{(O)}(1|v_2^{(1)}) \leq w_q < 1.
\end{equation}
However, one can prepare a qutrit system in the pure state $|v_1^{(1)}\rangle$ and map the boxes to the optimal orthonormal representation of the $k$-nested Clifton graph in Thm. \ref{thm:optq-nestClif} such that all the Hardy conditions are strictly satisfied while the quantum value evaluates to 
\begin{equation}
    P^{(q)}(1|v_2^{(1)}) = \left(\frac{k}{k+2}\right)^2.
\end{equation}
For any $O$ therefore, simply selecting a nesting depth $k > \frac{2 \sqrt{w_q}}{1- \sqrt{w_q}}$ implies that $P^{(O)}(1|v_2^{(1)}) <  P^{(q)}(1|v_2^{(1)})$ providing a Hardy-like test that state-dependently falsifies the $O$-valued noncontextual ontology.

We remark here that the above test makes the standard minimal assumptions for contextuality tests - that the observables measured obey the specified compatibility structure, that measuring one observable does not physically disturb or alter the system in a way that affects the subsequent outcome of a compatible observable. Methods have been devised to even relax these minimal assumptions and allow for an $\epsilon$-degree of dependence on the measurement context (see for example \cite{LR25, KDL15, Fyrillas23}). It would be interesting to pursue experimental realizations of the above tests incorporating these corrections and investigate the class of models that can be ruled out in physical realizations.

We also remark that while the above Hardy-type test can rule out general discrete noncontextual alternatives to the Born rule, in the literature other alternatives to the Born rule have also been considered and argued to lead to inconsistencies in other ways, see for instance \cite{Aaronson04}. For instance, one alternative to the probability of an outcome $|k\rangle$ in an orthonormal basis $\{|k' \rangle\}$ given a state $|\psi\rangle$ is \cite{GM17, Aaronson04}
\begin{eqnarray}
    P(k|\psi) = \frac{|\langle \psi | k \rangle|^{\alpha}}{\sum_{k'} |\langle \psi | k'\rangle|^{\alpha}} \quad \text{for} \; \alpha \neq 2.
\end{eqnarray}
In this case however observe that the probability assignment is explicitly contextual - the probability assigned to an outcome $|k \rangle$ depends explicitly upon the context $\{| k' \rangle\}$ in which it is measured owing to the denominator in the expression. As such, the Hardy test derived in this section does not rule out such alternatives (although they can be ruled out through other considerations \cite{GM17, Aaronson04})

\section{Applications}
\subsection{A stronger notion of Kochen-Specker contextuality}
\label{app:Sheaftheory-O}
In quantum information science, Kochen-Specker contextuality goes beyond being a philosophical quirk and is valued as a resource \cite{HZS+23, ABM17}behind a multitude of applications - providing the magic behind universal fault-tolerant computation via magic state distillation, certifying randomness and key in (semi)-device-independent cryptography, certifying quantum states via self-testing correlations and identifying scenarios with quantum communication complexity advantage. Viewed as a resource for such applications, it is important to quantify and understand the strongest manifestations of quantum contextuality.

A rigorous hierarchy of contextuality was formulated by Abramsky and Brandenburger in their sheaf-theoretic approach to contextuality in \cite{AB11}. In this framework, contextuality is not seen as simply a violation of a noncontextuality inequality, but more rigorously as a topological obstruction to the existence of a global section. Let us briefly recall the approach and then place our result as a stronger notion of contextuality than hitherto considered. 

We begin by formalising a physical contextuality experiment in terms of a measurement cover and the event presheaf following the approach of \cite{AB11}. 
\begin{definition}
    A measurement scenario is a triple $\langle X, A, \mathcal{M} \rangle$ where $X$ is a finite set of measurement labels (propositions/observables), $A$ is a finite set of measurement outcomes, $\mathcal{M}$ is a measurement cover, i.e., a family of subsets of $X$ such that $\cup_{C \in \mathcal{M}} C = X$. An element $C \in \mathcal{M}$ is a maximal context of jointly measurable, mutually exclusive propositions.
\end{definition}
\begin{definition}
    A presheaf $\mathcal{E}: \mathcal{P}(X)^{op} \to \text{Set}$ is defined on the powerset $\mathcal{P}(X)^{op}$ of $X$ ordered by inclusion. For any subset $U \subseteq X$, $\mathcal{E}(U) = A^{U}$ is the set of all outcome assignments (functions from $U$ to $A$). An element $s \in \mathcal{E}(U)$ is a local section. For any inclusion $V \subseteq U$, the restriction $\rho_{U,V}: \mathcal{E}(U) \to \mathcal{V}$ is the function restriction $\rho_{U,V}(s)=s|_{V}$. A global section is an element $g \in \mathcal{E}(X)$ representing a globally predefined hidden variable assignment for every observable in the system. 
\end{definition}
\begin{definition}
    Let $\mathcal{D}(S)$ denote the set of probability distributions over a finite set $S$. The event presheaf $\mathcal{E}$ is composed with the distribution functor $\mathcal{D}$ to yield a new presheaf $\mathcal{D}_{\mathcal{E}}(U) = \mathcal{D}(\mathcal{E}(U))$. For $V \subseteq U$ and $d \in \mathcal{D}_{\mathcal{E}}(U)$ the marginal distribution on $V$ is $d|_{V}(s') = \sum_{s \in \mathcal{E}(U), s|_{V} = s'} d(s)$.
\end{definition}
An empirical model is a family of probability distributions $e = \{e_{C}\}_{C \in \mathcal{M}}$ where each $e_C \in \mathcal{D}_{\mathcal{E}}(C)$ specifies the statistics for the context $C$. An empirical model $e$ is said to be no-signalling (aka no-disturbance) if for any two overlapping contexts $C_1, C_2 \in \mathcal{M}$, their marginals agree on the intersection (see App. \ref{sec:NCHV}):
\begin{eqnarray}
    e_{C_1}|_{C_1 \cap C_2} = e_{C_2}|_{C_1 \cap C_2}.
\end{eqnarray}
The no-disturbance requirement is thus simply that the family $\{e_C\}$ form a compatible family of local sections. Thus a no-disturbance model is a global section of the presheaf of probability distributions. An empirical model $e$ is non-contextual if there exists a global probability distribution $d \in \mathcal{D}_{\mathcal{E}}(X)$ such that for all $C \in \mathcal{M}$, $d|_{C} = e_{C}$. If no such global distribution exists, the model is contextual.   

Abramsky and Brandenburger then introduced a strict hierarchy in terms of probabilistic vs possibilistic vs strong contextuality. An empirical model $e$ is said to be possibilistically contextual if there exists no global section $g \in \mathcal{E}(X)$ that is consistent with the local supports. An empirical model $e$ is said to be strongly contextual if no local section in the support of any context can be extended to a consistent global section. In this framework, strong contextuality represents the absolute maximum deviation from classicality - in a strongly contextual model, the fraction of the behavior that can be explained by a classical hidden variable theory is zero (stated in other words, its contextual fraction is one \cite{ABM17}). Abramsky et al. \cite{ABKLM15} also demonstrated that strong contextuality manifests as a non-trivial cohomology class in the first cohomology group providing a purely homological witness for quantum paradoxes. In this paradigm, the quantum behavior maximally violating the CHSH inequality is contextual, but is not strongly contextual. In contrast, the Popescu-Rohrlich (PR) box is seen to be strongly contextual \cite{PR94, AB11} - there are no global sections compatible with its support. 

While the hierarchy is mathematically rigorous, it is fundamentally restricted to a $\{0,1\}$ ontological truth-set. Here we extend this to an even stronger form of contextuality (that we term strong $O$-valued contextuality) where we test empirical models $e$ against the set of $O$-valued models. 
\begin{definition}
Let $O = \{0, w_1, \ldots, w_q, 1\}$ with $0 < w_1 < \ldots < w_q < 1$. We define the $O$-valued event presheaf $\mathcal{E}_{O}: \mathcal{P}(X)^{op} \to \text{Set}$ over the base space $X$. For any subset $U \subseteq X$, the set of local sections $\mathcal{E}_{O}(U)$ consists of all $O$-valued assignments $s: U \to O$
satisfying the exclusivity and normalization conditions. Specifically for any maximal context $c \subseteq U$ we have that $\sum_{v \in c} s(v) \le 1$ and for any maximum context $c \subseteq U$ we have that $\sum_{v \in c} s(v) = 1$. 

For any inclusion $V \subseteq U$, the restriction map $\rho_{U,V}: \mathcal{E}_{O}(U) \to \mathcal{E}_{O}(V)$ is the standard function restriction $\rho_{U,V}(s) = s|_{V}$.
A global section $g \in \mathcal{E}_{O}(X)$ is a total function $g: X \to O$ such that its restriction to any context is a valid local section. 

For an empirical model $e$, the $O$-valued contextual fraction $CF_{O}(e)$ is the minimum weight of the empirical model that cannot be decomposed into a convex combination of $O$-valued global sections $g \in \mathcal{E}_{O}(X)$. 
\end{definition}

We finally arrive at the stronger notion of contextuality, that we term Strong $O$-valued Contextuality.
\begin{definition}
    An empirical model $e$ is $O$-valued contextual if $CF_{O}(e) > 0$. An empirical model $e$ is Strongly $O$-valued Contextual if no fraction of the empirical behavior can be explained by an $O$-valued global section in the presheaf $\mathcal{E}_{O}$, i.e., $\not \exists \; g \in \mathcal{E}_{O}(X) \; \text{s.t.} \; g(x) = e(x) \; \forall x \in X$, i.e., if $CF_{O}(e) = 1$.
\end{definition}
Let $e_{PR}$ denote the empirical model of the Popescu-Rohrlich box \cite{PR94}. While it exhibits standard Strong Contextuality ($CF_{\{0,1\}}(e_{PR}) = 1)$, we see that it fails to exhibit any $O$-valued contextuality for $O_1 = \{0,1/2,1\}$, that is there exists a single valid global section in $\mathcal{E}_{O_1}(X)$ and $CF_{O_1}(e_{PR}) = 0$. On the other hand, we see that the quantum empirical models arising from the $k$-nested Clifton graph exhibit $O$-valued contextuality for any $O = \{0, w_1, \ldots, w_q, 1\}$ such that $\frac{2 \sqrt{w_q}}{1-\sqrt{w_q}} < k$. 

The notion of strong $O$-valued contextuality extends the probabilistic-possibilistic-strong contextuality hierarchy of \cite{AB11}. While \cite{AB11} proved that contextuality is the obstruction of the Boolean presheaf $\mathcal{E}_{\{0,1\}}$, we have seen that quantum theory also evades the global sections for all finite-valued presheaves $\mathcal{E}_{O}$. 

\subsection{Extremal no-signalling behaviors and units of nonlocality}
Kochen-Specker contextuality is intimately connected with Bell nonlocality \cite{RW04, Liu2024}. In particular, it is known that for certain orthogonality graphs by distributing the contextual measurements across multiple, spatially separated subsystems one can convert KS contextuality into Bell nonlocality. To elaborate, in a KS scenario, one aims show that a non-contextual hidden-variable theory fails on a single system because measurement outcomes depend on which other compatible observables are measured at the same time. For certain configurations of observables, it is known that by giving each party a subset of these observables and measuring on an entangled states shared between the parties, one can force the system to exhibit nonlocality. In this process, the context-dependence within one system becomes the action-at-a-distance (nonlocality) between two separated systems, which then forms the bedrock of applications in quantum information theory to device-independent cryptography \cite{BarrettPRL, our-4, ZRLH22}. As such, the results presented in this paper for KS contextuality also have implications in Bell nonlocality. We highlight one of the consequences in this section. 

In an interesting set of papers in 2005, Barrett and Pironio \cite{BP05} and Jones and Masanes \cite{JM05} investigated the convex polytope of bipartite non-signalling correlations where two separated parties Alice and Bob choose from $d_x$ and $d_y$ measurement settings respectively, with each measurement yielding $d_a$ and $d_b$ measurement outcomes. In particular, the case of binary outcomes $d_a = d_b = 2$ was considered. By characterising the extreme points of the space of no-signalling correlations $P_{ab|xy}$ (the analog of the no-disturbance polytope), they revealed that every nonlocal extremal distribution inherently contains the structure of a Popescu-Rohrlich (PR) box \cite{PR94}. From an information-theoretic perspective, they proved that all such extremal nonlocal correlations are completely interconvertible under local operations and shared randomness. Consequently, the PR box was shown to serve as the fundamental unit of bipartite nonlocality in this regime, capable of simulating any two-outcome measurement performed on any bipartite quantum state. Specifically, the following theorem was shown.
\begin{theorem}[\cite{JM05, BP05}]
\label{thm:JM05}
    Let $\mathcal{P}_{NS}(d_x,2;d_y,2)$ be the convex polytope of bipartite nonsignaling correlations $P_{ab|xy}$ in the Bell scenario where the inputs $x \in \{0,\ldots, d_x - 1\}$ and $y \in \{0,\ldots, d_{y} - 1\}$ are drawn from alphabets of arbitrary finite cardinality, and the outcomes $a, b \in \{0,1\}$ are strictly binary. 
    
    For any extremal behavior $P^{ext} \in \mathcal{P}_{NS}(d_x,2;d_y,2)$, the marginal probability assignments for both Alice and Bob $P^{ext}_{a|x} \equiv \sum_b P_{ab|xy}$ and $P^{ext}_{b|y} \equiv \sum_a P_{ab|xy}$, are strictly constrained to the discrete semantic set:
    \begin{eqnarray}
        P^{ext}_{a|x}, \; P^{ext}_{b|y} \in \{0, 1/2, 1\} \quad \forall x,y,a,b.
    \end{eqnarray}
\end{theorem}
Jones and Masanes also computed the extremal points of the no-signalling polytope for the Bell scenario $(3,3;3,3)$ with three inputs and three outputs per player, and found an extremal point that had marginals in the set $\{0,1/4,1/2,1\}$ showing that going to the scenario with three outcomes evades the $\{0,1/2,1\}$ marginal characterisation of Theorem \ref{thm:JM05}. 
In an earlier counterpart result, \cite{BLMPPR05} had shown that for $d_x = d_y =2$ and $d_a = d_b = d$, the extremal points are strictly constrained to have marginals in the set $\{0, 1/d, 1\}$. The significance and motivation for these results in quantum information comes from the fact that just as treating the singlet state as a unit of entanglement underpins entanglement resource theory and motivated questions of asymptotic interconversion, multipartite scenarios etc., so too there is an analogous set of unanswered information-theoretic problems involving units of nonlocality.  We now leverage the results proven in this paper to show that if one goes beyond the cases $(2,d;2,d)$ and $(d_x,2; d_y,2)$ to the scenario of $3$ outcomes per player and general $d_x = d_y = d$ then there is no finite set to which the marginals of the extremal behaviors can be constrained for all $d$, i.e., the cardinality of the set containing the marginal probabilities of the extremal non-signalling boxes grows to an infinite set of dense rational values as $d$ grows. 

Specifically, we consider that Alice and Bob perform measurements given by the maximal cliques of size $3$ (triangles) of the $k$-nested Clifton graph $G_k$ obtaining one of three outcomes each. Given that there are $6k+1$ triangles once every maximal clique is completed to a triangle, this corresponds to the Bell scenario $(6k+1, 3; 6k+1,3)$ for $k \geq 1$. Let $\bar{G}_k$ denote the completion of the $k$-nested Clifton graph such that every maximal clique is a triangle. Consider the Bell experiment in which Alice and Bob select a measurement setting corresponding to a triangle $x, y \in \{C_1, C_2, \ldots, C_{6k+1}\}$. Their outcomes $a,b \in \{1,2,3\}$ corresponding to the mutually exclusive vertices within the chosen triangle. Let $\mathcal{P}_{NS}(6k+1,3;6k+1,3)$ denote the convex polytope of correlations $P_{ab|xy}$ satisfying non-negativity and normalization and the affine non-signalling constraints: $\sum_b P_{ab|xy} = \sum_b P_{ab|xy'} \equiv P_{a|x}$ and $\sum_a P_{ab|xy} = \sum_a P_{ab|x'y} \equiv P_{b|y}$.

\begin{theorem}
\label{thm:ext-ptsNS}
    Let $\mathcal{P}_{NS}(6k+1,3;6k+1,3)$ denote the convex polytope of bipartite nonsignaling correlations $P_{ab|xy}$ in the Bell senario in which Alice and Bob select a measurement setting corresponding to a maximum clique $x, y \in \{C_1, C_2, \ldots, C_{6k+1}\}$ in $\bar{G}_k$ and obtain outcomes $a, b \in \{1,2,3\}$. For increasing values of $k$, the marginals $P^{ext}_{a|x}$ and $P^{ext}_{b|y}$ of the extremal points of $\mathcal{P}_{NS}(6k+1,3;6k+1,3)$ populate a set of distinct rational values of increasing cardinality.
\end{theorem}
\begin{proof}
    Suppose that by contradiction the marginals of all extremal behaviors are constrained to a finite set $O$. We construct a no-signaling box that cannot be decomposed as a convex combination of such behaviors providing a contradiction. 

    The no-signaling box $P$ is constructed from the optimal orthogonal representation in Thm. \ref{thm:optq-nestClif} and the resulting quantum contextual behavior. Specifically let Alice's marginal behavior $P_{a|x}$ and Bob's marginal $P_{b|y}$ be constructed as the probabilities resulting from the projective measurements in \eqref{eq:optq-meq1} and \eqref{eq:optq-genm} on the state $|v_1^{(1)}\rangle = (0, \cos{\theta_k}, \sin{\theta_k})^T$ with $\theta_k = \frac{1}{2} \arcsin{\frac{k}{k+2}}$. Construct the no-signaling behavior $P$ with perfect synchronous correlations between Alice and Bob, specifically $P_{a|x} = P_{b|y}$ for $a=b, x=y$, $P_{ab|xy} = 0$ when $\langle v_{a,x} | v_{b,y} \rangle = 0$. Note that $P_{a|x} = P_{b|y} = 1$ for the input-output pair corresponding to $|v_1^{(1)}\rangle$ and $P_{a|x} = P_{b|y} = \left(\frac{k}{k+2}\right)^2$ for the input-output pair corresponding to $|v_2^{(1)}\rangle$. This gives a super-quantum nonsignaling correlation for which the marginal behavior is a consistent (no-disturbance) empirical model. 

    Now, suppose that the marginal behaviors are contained within a finite set $O = \{0,w_1,w_2, \ldots, w_q, 1\}$ with $0 < w_1 < w_2 < \ldots < w_q < 1$. Then as we have seen in App. \ref{app:HardyO-context}, under the condition that $f(v_1^{(1)}) = 1$ the probability corresponding to $|v_2^{(1)}\rangle$ is at most $w_q$. For $k = \lfloor \frac{2 \sqrt{w_q}}{1 - \sqrt{w_q}} \rfloor + 1$, we have that $w_q < \left(\frac{k}{k+2}\right)^2$ so that the no-signaling behavior $P^*$ cannot belong to the convex hull of behaviors with marginal probabilities in $O$. We have reached a contradiction, therefore, the marginals of the extremal points of $\mathcal{P}_{NS}(6k+1,3;6k+1,3)$ must populate an increasing set of distinct rational values growing infinitely with $k$.   
\end{proof}
Thm. \ref{thm:ext-ptsNS} establishes a characteristic of the extremal non-local points of the no-signalling polytope. It is known that quantum theory does not allow a realization of such extremal points \cite{SK25, RTHH16}. Extremal points are highly valuable in device-independent (DI) cryptography \cite{RTHH16, our-4} since they ensure that the devices held by honest parties are decoupled from that of any adversary. As such, one motivation for deriving state-independent contextuality proofs for ruling out $O$-valued models is to convert the corresponding quantum correlations to a bipartite Bell scenario via \cite{RW04}. This would realize quantum correlations that are close to extremal nonlocal vertices of the no-signalling polytope thereby providing a route to achieving security against this powerful class of adversaries in DI tasks. 

Viewed as a single system result, we also observe that the extremal points of the no-disturbance polytope $\mathcal{P}_{ND}(G_k)$ form a set of distinct rational values of increasing cardinality, and that the optimal quantum correlations in our tests get increasingly close to an extremal nonclassical point of the no-disturbance polytope as the number of measurements $6k+1$ grows. This provides a pathway towards achieving security of the corresponding semi-device-independent protocol for randomness expansion and amplification \cite{LR25} against the most general consistent adversaries, an application which we explicitly formulate in forthcoming work.

\subsection{A connection to finite many-valued logical models }
\label{app:many-valuedlogic-O}
In 1936, Birkhoff and von Neumann in their breakthrough paper "The logic of quantum mechanics" \cite{BvN36} argued that the structure of a set of dichotomic (`yes-no') propositions about properties of quantum objects, usually termed `quantum logic', differs from the structure of a set of such propositions pertaining to classical objects in which case it is a Boolean algebra. Since an experimental test of any such dichotomic proposition results in either a true or false outcome, quantum logic in the Birkhoff-von Neumann sense is generally treated as a $2$-valued non-classical logic. Classical theory is deterministic, any statement concerning future events is either true or false already at the moment of stating it. On the other hand, quantum theory is fundamentally indeterministic since the results of future experiments can be predicted only probabilistically, and this is a fundamental fact not caused by a lack of knowledge of some hidden parameters, i.e., quantum probabilities are 'ontic' not 'epistemic'. 

At the same time, other non-classical logics were studied,  among them being various kinds of many-valued logics. In particular, {\L}ukasiewicz \cite{Lukasiewicz70} a founding father of the theory of many-valued logics, argued that statements about future non-certain events belong to the domain of many-valued logic. He identified non-classical truth-values (different from $0$ or $1$) with probabilities that the considered statements will turn out to be true in future. For instance, if a projector is assigned an ontic probability of $1/2$, then in this framework the truth value of the proposition is $0.5$. Of course, once the projector is measured we see a result of true or false, which results in a collapse of the many-valued logic onto $2$-valued logic \cite{Pykacz94}. The question of whether quantum logic is related to the many-valued logic was investigated and settled in a neat result in \cite{Pykacz94} where a faithful isomorphism was established between Birkhoff-von Neumann quantum logic and a specific infinite-valued {\L}ukasiewicz logic. It must be noted that the assignments in many-valued logic are not natively treated as a probability but rather as an objective truth-value assigned to the event. To see the equivalence established in \cite{Pykacz94}, we elaborate with some mathematical details.  

The standard model of quantum logic \cite{Svozil98} is the family $\mathcal{L}(\mathcal{H})$ of closed linear subspaces of a Hilbert space or equivalently the family of operators of orthogonal projections onto these subspaces. These families are orthomodular lattices partially ordered, by the set-theoretic inclusion and by the relation $P_1 \leq P_2$ iff $P_1 P_2 = P_1$. Formally, by a quantum logic is meant an orthocomplemented $\sigma$-orthocomplete orthomodular poset \cite{Pykacz94}, i.e., a partially ordered set $L$ that contains the smallest element $0$ and the greatest element $I$ and an orthocomplementation map $\perp: L \to L$ satisfying the following conditions exists
\begin{itemize}
\item $(v^{\perp})^{\perp} = v$, 
\item $v \leq w \implies w^{\perp} \leq v^{\perp}$,
\item The greatest lower bound (meet) $v \cap v^{\perp}$ and the least upper bound (join) $v \vee v^{\perp}$ with respect to the give partial order exist in $L$ and $v \cap v^{\perp} = 0$ and $v \vee v^{\perp} = I$.
\end{itemize}
Furthermore, the $\sigma$-orthocompleteness condition holds:
\begin{itemize}
    \item If $v_i \leq v^{\perp}_j$ for $i \neq j$ (such elements are orthogonal and denoted $v_i \perp v_j$) then the join $\bigvee_i v_i$ exists in $L$.
    \item If $v \leq w$ then $w = v \vee (v^{\perp} \cap w) = v \vee (v \vee w^{\perp})^{\perp}$.
\end{itemize}
As per Birkhoff and von Neumann \cite{BvN36}, elements of a quantum logic represent experimentally verifiable dichotomic propositions about a physical system. Elements $v, w \in L$ are compatible if there exist in $L$ pairwise orthogonal elements $v_1, w_1, u$ such that $v = v_1 \vee u$ and $w = w_1 \vee u$. Compatibility of elements in a Hilbertian quantum logic $\mathcal{L}(\mathcal{H})$ means that the projectors are commuting. 

A probability measure (also called a state) on a quantum logic $L$ is a mapping $s: L \to [0,1]$ such that 
\begin{itemize}
    \item $s(I) = 1$,
    \item $s(\bigvee_i v_i) = \sum_i s(v_i)$ for any sequence of pairwise orthogonal elements of $L$. 
\end{itemize}
In a Hilbertian quantum logic, by Gleason's Theorem these probability measures are in one-to-one correspondence with density operators on the Hilbert space. 

Maczy\'{n}ski \cite{Maczynski73} proved a functional representation theorem that formed the bridge by which one can move from Birkhoff-von Neumann quantum logic in its order-theoretic form above to its fuzzy set or many-valued representation. Pykacz \cite{Pykacz94} built upon this theorem to write BvN quantum logic as infinite-valued logic. In this framework, set-theoretic complementation is related to logical negation, intersection of sets is related to conjunction and union of sets is related to disjunction as per usual notions. In the case of fuzzy sets and many-valued logics though, the specific forms of these expressions depend on the model of negation, conjunction and disjunction. We adopt the original {\L}ukasiewicz negation: 
\begin{equation}
\label{eq:Lneg}
    \tau(\neg p) = 1 - \tau(p)
\end{equation} 
where $\tau(p) \in [0,1]$ is the truth value of proposition $p$, and as per {\L}ukasiewicz this number is identified with the probability that the proposition $p$ will turn out to be true when a suitable experiment is performed. The conjunction and disjunction are the logical counterparts of the fuzzy set operations and are given as
\begin{eqnarray}
\label{eq:Lcondis}
    \tau(p \sqcap q) &=& \max\{ \tau(p) + \tau(q) - 1, 0 \}, \nonumber \\
    \tau(p \sqcup q) &=& \min\{\tau(p) + \tau(q), 1 \}.
\end{eqnarray}
With these, we are in position to state the equivalence between quantum logic and infinite many-valued logic established in \cite{Pykacz94, Maczynski73} as
\begin{theorem}[\cite{Pykacz94, Maczynski73}]
\label{thm:Pykacz}
    Any quantum logic $L$ with an ordering set of probability measures $S$ can be isomorphically represented as a family $\Lambda(S)$ of propositional functions defined on $S$ and satisfying the following conditions:
    \begin{enumerate}
        \item $\Lambda(S)$ contains the always-false propositional function $f$.
        \item $\Lambda(S)$ is closed with respect to {\L}ukasiewicz negation \eqref{eq:Lneg}.
        \item $\Lambda(S)$ is closed with respect to {\L}ukasiewicz disjunction of pairwise exclusive propositional functions, i.e., if $v(\cdot)_i \sqcap v(\cdot)_j = f$ for all $i \neq j$ then $\bigsqcup_i v(\cdot)_i \in \Lambda(S)$.
        \item The always-false propositional function $f$ is the only propositional function in $\Lambda(S)$ that excludes itself, i.e., for any $v(\cdot) \in \Lambda(S)$, if $v(\cdot) \sqcap v(\cdot) = f$, then $v(\cdot) = f$.
    \end{enumerate}
    Conversely, any family of many-valued propositional functions defined on a common domain $D$ and satisfying the conditions above is a quantum logic in the BvN sense when we identify propositional functions that assume the same truth value for every argument in $D$. This family is partially ordered as $v(\cdot) \leq w(\cdot)$ iff $\tau(v(x)) \leq \tau(w(x))$ for all $x \in D$ with {\L}ukasiewicz negation as orthocomplementation, orthogonality of elements coinciding with their exclusivity, and an ordering set of probability measures being generated according to $s_x(v(\cdot)) = \tau(v(x))$ for all $x \in D$.
\end{theorem}
In words, we have the following equivalence. On the left-hand side (quantum logic), we have projectors $P$ that elements of the quantum logic and represent dichotomic experimental propositions, and density operators $\rho$. This is translated into many-valued logic with each projector $P$ being mapped to a propositional function $\tilde{P}(\cdot): S \to [0,1]$ where $S$ parametrizes the states. The truth value $\tilde{P}(x)$ at a state $x$ corresponds to the probability $s_x(P) = \text{Tr}[\rho_x P]$. The full family of such propositional functions satisfies the conditions stated in Thm. \ref{thm:Pykacz}, and {\L}ukasiewicz operations on these functions mirror the quantum logic operations. The conjunction and disjunction are partial, i.e., one can only conjoin/disjoin propositions when the corresponding projectors commute (are orthogonal in the abstract logic). This reformulation shows that quantum propositions about future measurement outcomes are naturally many-valued and state-dependent, with truth degrees in $[0,1]$ interpreted as objective probabilities of the proposition becoming true upon measurement. This equivalence theorem therefore provides a rigorous many-valued logical semantics for quantum theory.

As seen above, it is recognized in the literature that finite-valued logics are mathematically untenable for quantum theory as a corollary of Gleason's theorem. However, as we have seen Gleason's theorem relies intrinsically on the continuity of the entire sphere, and it is preferable to have an experimentally amenable test that queries only a strictly finite subset of propositions (a finite partial Boolean algebra) to rule out finite-valued logics. We therefore formalize the statement that the finite partial Boolean algebra arising from the nested Clifton graph and the quantum violation of the corresponding noncontextuality inequality serves to rule out finite-valued logical models.

Let $G_k = (V_k, E_k)$ denote the nested Clifton graph seen as the exclusivity graph of the finite experimental events corresponding to the rank-$1$ projectors with orthogonal projections corresponding to mutually exclusive propositions. The graph generates a finite partial Boolean algebra $\mathcal{B}(G) \subset \mathcal{L}(\mathcal{H})$. Let $O$ be a finite, ordered set of numbers representing the allowed ontological degrees of truth, $O = \{0, w_1, w_2, \ldots, w_q, 1\}$ with $0 < w_1 < w_2 < \ldots < w_q < 1$. 

\begin{definition}
    A finite-valued logical model is a mapping $s: \mathcal{B}(G) \to O$,  subject to the exclusivity and normalization conditions \eqref{eq:Lneg}, \eqref{eq:Lcondis}. For any maximal clique of mutually exclusive propositions $\{P_1, \ldots, P_d\}$ in $G$, the valuation is finitely additive:
    \begin{eqnarray}
        s\left(\bigvee_{i=1}^d P_i \right) = \sum_{i=1}^n s(P_i) \leq 1.
    \end{eqnarray}
\end{definition}
From Thms. \ref{thm:NestClif-NDval} and \ref{thm:optq-nestClif} in App. \ref{app:HardyO-context}, we see that for any such $O$ with a maximum non-certain truth value $w_q < 1$, one can find a $k$-nested Clifton graph corresponding to a threshold $k_{th}$ given by the smallest positive integer such that
\begin{equation}
    \left(\frac{k}{k+1} \right)^2 > w_q \implies k > \frac{2 \sqrt{w_q}}{1-\sqrt{w_q}},
\end{equation}
for which under $f(v_1^{(1)}) = 1$ we have the quantum value being strictly greater than the truth value $w_q$. That is, for $k_{th} = \lceil \frac{2 \sqrt{w_q}}{1-\sqrt{w_q}} \rceil$, we are either forced to violate the orthomodularity of the $6k+2$ propositions or abandon the finite-valued logic. 

\subsubsection{Vorob'ev theorem and the connection with the Treewidth of the Orthogonality Graph}
The Fine Theorem \cite{Fine82} and Vorob'ev Theorem \cite{Vorobev62} are seminal results in quantum foundations that address the marginal problem: given a set of marginal probability distributions, under what conditions do they come from a single, global joint probability distribution? In the Bell scenario, the Fine Theorem shows that observables admits a classical joint probability distribution (a local hidden-variable model) if and only if it satisfies a complete set of Bell inequalities. In a general contextuality scenario, the Vorob'ev Theorem \cite{Vorobev62} states that a system of observables admits a classical joint probability distribution (i.e., the measurements do not generate contextuality regardless of the quantum state on which they are performed) if and only if its compatibility graph is chordal \cite{RSKK12}. Here the compatibility graph is the undirected graph where vertices represent observables and an edge between two observables denotes commutation, and a chordal graph is a graph where every cycle of length $4$ or greater has a chord - an edge connecting two non-consecutive vertices of the cycle. 

In this context, it is natural to investigate properties of the $k$-nested Clifton graph $G_k$ which rules out $O$-valued models (the analog of the Fine theorem in this case is that the behavior lies within the convex polytope given by the convex hull of all $O$-valued probability assignments). Of special interest is the treewidth of the graph \cite{Diestel05}, which measures how close the graph is to being a tree (and hence fully noncontextual), we recall its formal definition in terms of chordal completions here.
\begin{definition}[Treewidth via Chordal Completion]
Let $G = (V, E)$ be a simple, undirected graph. A chordal graph $H = (V, E_H)$ is a chordal completion of $G$ if $E \subseteq E_H$. Let $\mathcal{H}(G)$ denote the set of all valid chordal completions of $G$. The treewidth of $G$, denoted by $\mathrm{tw}(G)$, is defined as:
\begin{equation}
\mathrm{tw}(G) = \min_{H \in \mathcal{H}(G)} \omega(H) - 1,
\end{equation}
where $\omega(H)$ is the clique number of $H$.
\begin{lemma}
    The Treewidth of the $k$-nested Clifton graph $G_k$ is $3$, i.e., $\mathrm{tw}(G_k) = 3$ for all $k$. 
\end{lemma}
\begin{proof}
We first focus on $G_1$ and  consider the graph formed by suppressing the degree-2 outer vertices $v_1^{(1)}$ and $v_2^{(1)}$, and adding a virtual edge between their neighbors: $\{v_3^{(1)}, v_6^{(1)}\}$ and $\{v_4^{(1)}, v_7^{(1)}\}$. The resulting 6-vertex core is isomorphic to the triangular prism graph, which is seen to have a treewidth of $3$ by explicit tree decomposition. Suppressing degree-2 vertices is known to not change the treewidth of graphs where $\mathrm{tw} \ge 2$, so that $\mathrm{tw}(G_1) = 3$.

We now proceed to consider the general graph $G_k$. The operation of gluing layer $m$ to layer $m+1$ at vertices $v_5$ and $v_8$ is a graph-theoretic clique-sum when we assume a virtual edge between the gluing vertices. A fundamental theorem of tree decompositions \cite{Lovasz06} states that the treewidth of a clique-sum of two graphs is simply the maximum of the treewdiths of the individual graphs (the individual layers which are all isomorphic to the Clifton bug in our case). Therefore, $\mathrm{tw}(G_k) = \max(\mathrm{tw}(C^{(1)}), \dots, \mathrm{tw}(C^{(k)})) = 3$.
\end{proof}

\end{definition}
In this regard, it is worth noting that Vorob'ev theorem states that graphs with treewidth equal to one are all noncontextual (the corresponding behaviors can be explained by $O = \{0,1\}$-valued models and their convex mixtures). One might expect that graphs with increasing treewidth would be required to rule out more general $O$-valued models, however we observe that the $k$-nested Clifton graph has only a small constant treewidth defying this expectation. 

\section{Gleason Theorem rules out consistent probability assignments from finite subsets of $[0,1] \setminus \{1/d\}$}
\label{sec:Gleason}
Here, we show the explicit connection between the Gleason theorem \cite{Gleason57} and the Kochen-Specker theorem \cite{KS67} via Logical Compactness \cite{Enderton01}. Informally, Gleason's theorem establishes that under the constraint of orthogonal additivity, the Born rule is forced, i.e., every countably additive state has the form $\mu(|\psi\rangle) = \langle \psi | \rho | \psi \rangle$ for a unique density operator $\rho$ and any unit vector $|\psi\rangle$. This map $|\psi \rangle \to \langle \psi | \rho | \psi \rangle$ is continuous and descends to the projective space which is connected. A $\{0,1\}$-valued continuous function on the connected space is just a constant. One can readily see that a constant $0$ assignment to all unit vectors implies that the sum over any orthonormal basis is $0 \neq 1$, and a constant $1$ assignment implies that the sum over any orthonormal basis is $\text{dim}(\mathcal{H}) > 1$. In either case, we obtain a contradiction so that no $\{0,1\}$-valued assignment is possible when one considers the entire set of unit vectors.  

Here, we formulate the probability assignments from finite $\mathcal{O} \subset [0,1]\setminus \{1/d\}$ as a satisfiability problem and use this to show analogously that the Gleason theorem rules out consistent probability assignments from $\mathcal{O}$ and consequently convex combinations thereof. 

Let \(\mathcal{H}\) be a complex Hilbert space with \(\dim(\mathcal{H}) = d \ge 3\). The projective space \(\mathbb{P}(\mathcal{H})\) is the set of all 1-dimensional linear subspaces (rays) of \(\mathcal{H}\). Equivalently, \(\mathbb{P}(\mathcal{H})\) is bijectively mapped to the set of all rank-1 orthogonal projection operators on \(\mathcal{H}\):
\[\mathbb{P}(\mathcal{H}) \cong \{ P \in \mathcal{B}(\mathcal{H}) \mid P = P^\dagger = P^2, \, \text{Tr}(P) = 1 \},\]
where \(P = \vert{}\psi\rangle\langle\psi\vert{}\) for \(\vert{}\psi\rangle \in \mathcal{S}(\mathcal{H})\). \(\mathbb{P}(\mathcal{H})\) is equipped with the topology induced by the operator norm \(\Vert{}P - Q\Vert{}\). Under this topology, \(\mathbb{P}(\mathcal{H})\) is a compact and connected topological space.

\begin{definition}
A full orthonormal frame (or basis) \(\mathcal{F}(\mathcal{H})\) is a subset of \(\mathbb{P}(\mathcal{H})\) containing exactly \(d\) elements:
\[B = \{P_1, P_2, \dots, P_d\} \subset \mathbb{P}(\mathcal{H}),\]
such that the operators are mutually orthogonal and sum to the identity operator:
$P_i P_j = \delta_{ij} P_i \quad \forall i,j \in \{1, \dots, d\}$ and  $\sum_{i=1}^d P_i = I$. 
The set of all such full orthonormal frames is denoted by \(\mathcal{F}(\mathcal{H})\).
\end{definition}

\begin{definition}
Let \(\Omega \subseteq \mathbb{R}\). A function \(f: \mathbb{P}(\mathcal{H}) \to \Omega\) is defined as a frame function of weight \(W\) if for every full orthonormal frame \(B \in \mathcal{F}(\mathcal{H})\), the sum of the function values over the elements of \(B\) equals a constant \(W \in \mathbb{R}\):
\[\sum_{P \in B} f(P) = W \quad \forall B \in \mathcal{F}(\mathcal{H}).\]
\end{definition}
Gleason's Theorem \cite{Gleason57} is the statement that in dimensions greater than two, essentially the only valid frame functions arise via the Born rule. Formally, it is stated as follows.
\begin{theorem}[Gleason's Theorem \cite{Gleason57}]\label{thm:gleason}
Let \(\mathcal{H}\) be a real or complex Hilbert space of \(\dim(\mathcal{H}) \ge 3\). Let \(f: \mathbb{P}(\mathcal{H}) \to [0, \infty)\) be a non-negative frame function of weight \(W\). Then there exists a unique positive semi-definite, self-adjoint, trace-class operator \(\rho\) acting on \(\mathcal{H}\) such that:
\[f(P) = \text{Tr}(\rho P) \quad \forall P \in \mathbb{P}(\mathcal{H}),\]
where the trace of the operator satisfies \(\text{Tr}(\rho) = W\).
\end{theorem}

We will also use the Logical Compactness theorem \cite{Enderton01} to infer the existence of a finite subset of rays that rule out probability assignments from a finite $\mathcal{O}$.
\begin{theorem}[Logical Compactness Theorem]\label{thm:compactness}
Let \(\Sigma\) be a set of sentences in a propositional logic language \(\mathcal{L}\). \(\Sigma\) is satisfiable if and only if every finite subset \(\Sigma_0 \subseteq \Sigma\) is satisfiable.
\end{theorem}

We can now state the theorem formally as follows.
\begin{theorem}
\label{thm:GleasGenKS}
Let \(\mathcal{H}\) be a complex Hilbert space with \(\dim(\mathcal{H}) = d \ge 3\). Let \(\mathcal{O}\) be a strictly finite subset of the real unit interval \([0,1]\) that explicitly excludes the value \(\frac{1}{d}\):
\[\mathcal{O} \subset [0, 1] \setminus \left\{ \frac{1}{d} \right\}, \quad \vert{}\mathcal{O}\vert{} < \infty.\]
There exists a finite subset \(K \subset \mathbb{P}(\mathcal{H})\) such that no partial function \(\nu_K: K \to \mathcal{O}\) satisfies \(\sum_{P \in B} \nu_K(P) = 1\) for all full orthonormal frames \(B \in \mathcal{F}(\mathcal{H})\) where \(B \subseteq K\).
\end{theorem}

\begin{proof}
Assume by contradiction that there exists a global context-independent assignment function \(\nu: \mathbb{P}(\mathcal{H}) \to \mathcal{O}\) that satisfies the frame condition of weight \(W=1\):
\begin{equation}\label{eq:frame_O}
\sum_{P \in B} \nu(P) = 1 \quad \forall B \in \mathcal{F}(\mathcal{H}).
\end{equation}
By definition, \(\nu\) is a frame function. Because its codomain is \(\mathcal{O} \subset [0, 1] \subset [0, \infty)\), \(\nu\) is a non-negative frame function. Since \(\dim(\mathcal{H}) \ge 3\), Gleason's Theorem \(\ref{thm:gleason}\) strictly guarantees the existence of a unique density operator \(\rho \ge 0\) with \(\text{Tr}(\rho) = 1\) such that:
\begin{equation}\label{eq:born_map}
\nu(P) = \text{Tr}(\rho P) \quad \forall P \in \mathbb{P}(\mathcal{H}).
\end{equation}
The map \(\Phi: \mathbb{P}(\mathcal{H}) \to \mathbb{R}\) defined by \(\Phi(P) = \text{Tr}(\rho P)\) is a linear trace function. Because the trace function is continuous with respect to the operator norm topology on \(\mathbb{B}(\mathcal{H})\), \(\Phi\) is a continuous map. A fundamental property of continuous mappings is that the image of a connected space under a continuous function must be a connected subset of the codomain. Therefore, \(\text{Im}(\nu) = \Phi(\mathbb{P}(\mathcal{H}))\) must be a connected subset of \(\mathbb{R}\). Since \(\mathcal{O}\) is a strictly finite subset of \([0,1] \setminus \{1/d\}\), the subspace topology it inherits from the standard topology of \(\mathbb{R}\) is the discrete topology. The only connected subsets of a discrete topological space are singletons.
Thus, the image of \(\nu\) must be a singleton \(\{c\}\) for some constant \(c \in \mathcal{O}\):
\[\nu(P) = c \quad \forall P \in \mathbb{P}(\mathcal{H}).\]
For an arbitrary full orthonormal frame \(B \in \mathcal{F}(\mathcal{H})\), we obtain:
\[\sum_{P \in B} \nu(P) = \sum_{P \in B} c = 1 \implies c = d^{-1}.\]
However, this directly contradicts the condition that \(\mathcal{O} \subset [0,1] \setminus \{\frac{1}{d}\}\). Therefore, no global context-independent assignment \(\nu: \mathbb{P}(\mathcal{H}) \to \mathcal{O}\) can exist.

To get to a finite set of vectors, we model the existence of an \(\mathcal{O}\)-valued assignment as a satisfiability problem in propositional logic. Let the finite set of valid values be indexed as \(\mathcal{O} = \{r_1, r_2, \dots, r_m\}\), where \(m = \vert{}\mathcal{O}\vert{}\). Construct a formal propositional language \(\mathcal{L}\) where the propositional variables are indexed by both the projection operators and the values in \(\mathcal{O}\):
\[\text{Var}(\mathcal{L}) = \{ X_{P, r} \mid P \in \mathbb{P}(\mathcal{H}), \, r \in \mathcal{O} \}. \]
The proposition \(X_{P, r}\) is interpreted semantically as \(\nu(P) = r\).

Let \(\Sigma\) be the infinite set of sentences over \(\mathcal{L}\) defined by the following structural rules:
\begin{enumerate}
    \item For every projector \(P \in \mathbb{P}(\mathcal{H})\), \(P\) cannot be assigned more than one value from \(\mathcal{O}\): $\bigwedge_{P \in \mathbb{P}(\mathcal{H})} \;\; \bigwedge_{r_i \neq r_j \in \mathcal{O}} \neg (X_{P, r_i} \wedge X_{P, r_j})$.
    
    \item For every projector \(P \in \mathbb{P}(\mathcal{H})\), \(P\) must be assigned at least one value from \(\mathcal{O}\): $ \bigwedge_{P \in \mathbb{P}(\mathcal{H})} \left( \bigvee_{r \in \mathcal{O}} X_{P, r} \right)$.

    \item For every full orthonormal frame \(B = \{P_1, P_2, \dots, P_d\} \in \mathcal{F}(\mathcal{H})\), the assigned values must sum to 1. Let \(\mathcal{S}_B\) be the finite set of all valid \(d\)-tuples of weights that satisfy the frame condition:
    $\mathcal{S}_B = \left\{ (a_1, a_2, \dots, a_d) \in \mathcal{O}^d \;\middle\vert{}\; \sum_{i=1}^d a_i = 1 \right\}$.
    The assigned values for the elements of \(B\) must match one of these valid tuples:
    \begin{equation}\label{eq:ax_frame}
    \bigvee_{(a_1, \dots, a_d) \in \mathcal{S}_B} \left( \bigwedge_{i=1}^d X_{P_i, a_i} \right).
    \end{equation}
\end{enumerate}
A truth valuation \(v: \text{Var}(\mathcal{L}) \to \{\text{True}, \text{False}\}\) satisfies the infinite set of sentences \(\Sigma\) if and only if it defines a consistent, global frame function mapping \(\mathbb{P}(\mathcal{H}) \to \mathcal{O}\) of weight 1. As we have seen above, no such global function can exist. Thus, the infinite theory \(\Sigma\) is strictly unsatisfiable. By the contrapositive of the Logical Compactness Theorem \ref{thm:compactness}, since the infinite set of sentences \(\Sigma\) is unsatisfiable, there must exist a finite subset of sentences \(\Sigma_0 \subset \Sigma\) that is also unsatisfiable.
Let $K \subset \mathbb{P}(\mathcal{H})$ be the finite set of projection operators whose corresponding propositional variables 
$X_{P,r}$ appear inside the sentences of the finite sub-theory $\Sigma_0$. Since $\Sigma_0$ is finite and unsatisfiable, it is logically impossible to map the elements of the finite set of projections $K$ to the finite set of values $\mathcal{O}$ 
 without violating the sum constraints for the finite collection of frames represented in $\Sigma_0$. This finite set 
 constitutes a Generalized Kochen-Specker configuration for outcome probability assignments from $\mathcal{O}$. 
 \end{proof}

While Theorem \ref{thm:GleasGenKS} rules out consistent probability assignments from the finite sets $\mathcal{O}$, it does so in a non-constructive manner. For experimental tests to rule out the corresponding theories, the Hardy-type tests formulated in the current paper are more relevant.

As an aside, also consider the case that $\mathcal{O} \subset [0,1] \setminus \{1/d\}$ is not a finite subset but a dense subset such as $\mathbb{Q} \cap [0,1] \setminus \{1/d\}$. Let \(\tilde{\mathcal{O}}\) be a dense subset of the real unit interval \([0, 1]\) that excludes the value \(\frac{1}{d}\). That is, the closure of \(\tilde{\mathcal{O}}\) satisfies $\overline{\tilde{\mathcal{O}}} = [0, 1], \quad \text{and} \quad \frac{1}{d} \notin \tilde{\mathcal{O}}$. 
In this case the logical compactness theorem cannot be applied because the length of the clauses becomes infinite. Infinitary logics do not possess the compactness property \cite{Enderton01}, an infinite set of sentences can be unsatisfiable even if every single one of its finite sub-theories is perfectly satisfiable. The non-existence of probability assignments from such an $\tilde{\mathcal{O}}$ cannot be therefore reduced to a finite configuration of projections via logical compactness. Nevertheless one can show the following via the same continuity argument, since the trace is a continuous map and the projective space is connected, the image cannot lie strictly inside a disconnected set.   .
\begin{prop}
Let \(\mathcal{H}\) be a Hilbert space with \(\dim(\mathcal{H}) = d \ge 3\). If \(\mathcal{O}\) is a dense subset of \([0, 1] \setminus \{\frac{1}{d}\}\), then there exists no global context-independent assignment \(\nu: \mathbb{P}(\mathcal{H}) \to \mathcal{O}\) satisfying $\sum_{P \in B} \nu(P) = 1$ for all full orthonormal frames $B \subset K$.     
\end{prop}
 
\section{Generalized Kochen-Specker theorem to rule out $O$-valued noncontextual models}
\label{app:GenKS}
Thus far, we have seen a family of experimentally feasible contextual tests to rule out arbitrary $O$-valued noncontextual models. However these tests are state-dependent being of the Hardy-type, i.e., they check that if $f(v_1^{(1)}) = 1$ then $f(v_2^{(1)}) \le w_q$ in an $O$-valued theory, while $f^{(q)}(v_2^{(1)}) = \left(\frac{k}{k+2}\right)^2$ in quantum theory. The test thus requires measurement on the specific state $|v_1^{(1)}\rangle$ and is state-dependent. 

On the other hand, the KS theorem is a state-independent test, i.e., the statistics obtained by measuring a set of projectors proving the KS theorem cannot be explained by an $O = \{0,1\}$-valued model irrespective of the state on which they are measured. It is thus highly desirable to have an analogous generalisation of the KS theorem that rules out $O$-valued models beyond the Boolean case. A first initial result in this direction was provided by us in \cite{Ravi24} for the case $O = \{0, p, 1-p, 1\}$ for $0 \leq p \leq 1/2$ and $p \neq 1/3$ (which included the interesting case $O = \{0,1/2,1\}$. Here we provide a KS theorem for a more general class of models which includes the result of \cite{Ravi24} as a special case.

Let $\Delta^2 \subset \mathbb{R}^3$ denote the standard 2-simplex, defined as:
\[
\Delta^2 = \left\{ (x, y, z) \in \mathbb{R}^3 \mid x \ge 0, y \ge 0, z \ge 0 \text{ and } x + y + z = 1 \right\}.
\]
The boundary of $\Delta^2$ is the set of points where at least one coordinate vanishes:
\[
\partial \Delta^2 = \left\{ (x, y, z) \in \Delta^2 \mid x=0 \text{ or } y=0 \text{ or } z=0 \right\}.
\]
The interior of the simplex consists of all strictly positive probability vectors:
\[
\text{int}(\Delta^2) = \Delta^2 \setminus \partial \Delta^2 = \left\{ (x, y, z) \in \mathbb{R}^3 \mid x > 0, y > 0, z > 0 \text{ and } x + y + z = 1 \right\}.
\]
The class of $O$-valued models for which we can prove a state-independent theorem is defined as follows. 
\begin{definition}
\label{def:Ostar}
    Let $O^* = \{0, w_1, w_2, \ldots, w_q, 1\}$ be such that the intersection of $O^3$ with the $2$-simplex is entirely contained within its boundary, i.e., $(O^*)^3 \cap \Delta^2 \subseteq \partial \Delta^2$ or equivalently $(O^*)^3 \cap \text{int}(\Delta^2) = \emptyset$.
\end{definition}
To elaborate, the set of $O$-valued models we consider in this section are the $O^*$ for which every probability assignment $(p_1,p_2,p_3) \in {O^*}^3$ satisfying the normalization condition $p_1 + p_2 + p_3 = 1$ contains at least one zero. We remark that this is a strict generalization of the models considered in \cite{Ravi24} since while $O^* = \{0, p, 1-p, 1\}$ satisfies the condition, there are a large class of other models that also fall under the condition, such as for instance $O^* = \{0, 1/m, 1\}$ for $m \geq 2$. The proof of the following theorem is constructive as opposed to the non-constructive proof obtained via the Gleason theorem (see App. \ref{sec:Gleason}). The proof is obtained via several improvements to the proof structure from our \cite{Ravi24}. 

\begin{theorem}
    Let $O^*$ be as in Def. \ref{def:Ostar}. There exists a finite set of rays $V_{O^*} \subset \mathbb{C}^3$ that admits no global $O^*$-valued assignment. 
\end{theorem}
\begin{proof}
Let $O^* = \{0, w_1, w_2,\ldots, w_q,1\}$ be given as in Def. \ref{def:Ostar}. The proof consists of three steps each of which is defined by a specific gadget consisting of a finite set of vectors. In the first step, we build upon the nested Clifton graph construction and present a gadget that rules out a `one-value' for any given vector in an $O^*$-valued assignment. This step is used to reduce the problem to the case of outcome probability assignments in $O^*\setminus \{1\}$ since a finite vector set that does not admit assignments from $O^* \setminus \{1\}$ can be augmented with finite gadgets from step one for each of the vectors in the set. In the second step, we construct a gadget that admits assignments from $O^* \setminus \{1\}$ but such that in any such assignment it is necessarily that three linearly independent vectors are all assigned value $0$. In the third step, we construct a gadget that admits assignments from $O^* \setminus \{1\}$ but such that in any such assignment it is necessarily the case that the three linearly independent vectors from step two cannot all be assigned value $0$. Taking the union of the finite vector sets from steps two and three gives a finite set that does not admit outcome assignments from $O^* \setminus \{1\}$. Augmenting this finite set $S$ with a gadget from step one for each of its vertices proves the final no-go theorem. 

\textit{Step one:} Accordingly, we first proceed to construct the gadget for step one in terms of the following lemma. 
\begin{lemma}
\label{lem:stepone}
    Let $|u \rangle \in \mathbb{C}^3$ be an arbitrary fixed vector. There is a finite set of vectors $V_{u} \subset \mathbb{C}^3$ with $|u \rangle \in V_{u}$ such that for any outcome assignment $f: V_{u} \rightarrow O^*$, it holds that $f(u) \neq 1$. 
\end{lemma}
\begin{proof}
    The proof of this lemma uses the nested Clifton graph of App. \ref{app:HardyO-context} and a further construction providing a proof of Pitowsky's Logical Indeterminacy Principle from \cite{Pitowsky98, HP04}. 

    We first consider the nested Clifton graph for a suitable level of nesting $k$ depending on $O^*$ shown in Fig. \ref{fig:clifton_nested}. The property of interest is that in any $O^*$-valued assignment $f$, it holds that $f(v_1) = 1 \implies f(v_2) < 1$ (see Prop. \ref{thm:NestClif-NDval}). The nested Clifton graph will be depicted by a dashed edge in our construction signifying this property. 

    Secondly, consider the gadget shown in Fig. \ref{fig:stepone-1}. In this gadget we have that if $f(w_1) = 1$ then by the exclusivity constraint $f(w_3) = 0$ and by the property of the dashed edge (the nested Clifton graph from above) $f(w_2) < 1$ and $f(w_4) < 1$. Since the completeness condition imposes that $f(w_2) + f(w_3) + f(w_4) = 1$ we have that $0 < f(w_2) < 1$ (as also $0 < f(w_4) < 1$). We denote this gadget by a directed dashed edge in our construction signifying this property. At this point, it is also worth remarking that this gadget was used by us to provide a simpler proof of Pitowsky's logical indeterminacy principle in \cite{RRHS+20}. 

    Our proof of the lemma is completed by the gadget shown in Fig. \ref{fig:stepone-2}. Suppose that $f(u_1) = 1$ in an $O^*$-valued assignment so that by exclusivity $f(u'_1) = 0$. We then have by completeness that $f(u_2) + f(u'_2) =1$ and $f(u'_2) + f(u'_3) + f(u_3) = 1$ giving $f(u_2) = f(u'_3) + f(u_3)$. By the property of the directed dashed edge $\{u_ 1, u'_3\}$ (implying $0< f(u'_3) < 1$) we obtain that $f(u_2) > f(u_3)$ (note the strict inequality). By the property of the dashed edge $\{u_1, u_2\}$ we have that $f(u_2) < 1$ implying that in an $O^*$-valued assignment there holds $f(u_2) \leq w_q$. Taken together this implies that $f(u_3) \leq w_{q-1}$. Call this gadget $\Gamma(u_1, u_2, u_3)$, it has the property that $f(u_1) = 1$ implies $f(u_2) \leq w_q$ and $f(u_3) \leq w_{q-1}$. Now consider a similar gadget with $\Gamma(u_1, u_3, u_4)$ giving $f(u_4) \leq w_{q-2}$. Taking the union $\bigcup_{t=1}^{q+2} \Gamma(u_1, u_t, u_{t+1})$ gives that when $f(u_1) = 1$, there holds $f(u_{q+3}) < 0$ which yields a contradiction.
Therefore, the gadget graph we have constructed does not admit an $O^*$-valued outcome assignment with $f(u_1) = 1$ proving the statement.
\end{proof}

\begin{figure}[htbp]
\centering
\begin{minipage}{.48\textwidth}
  \centering
  \includegraphics[width=.8\linewidth]{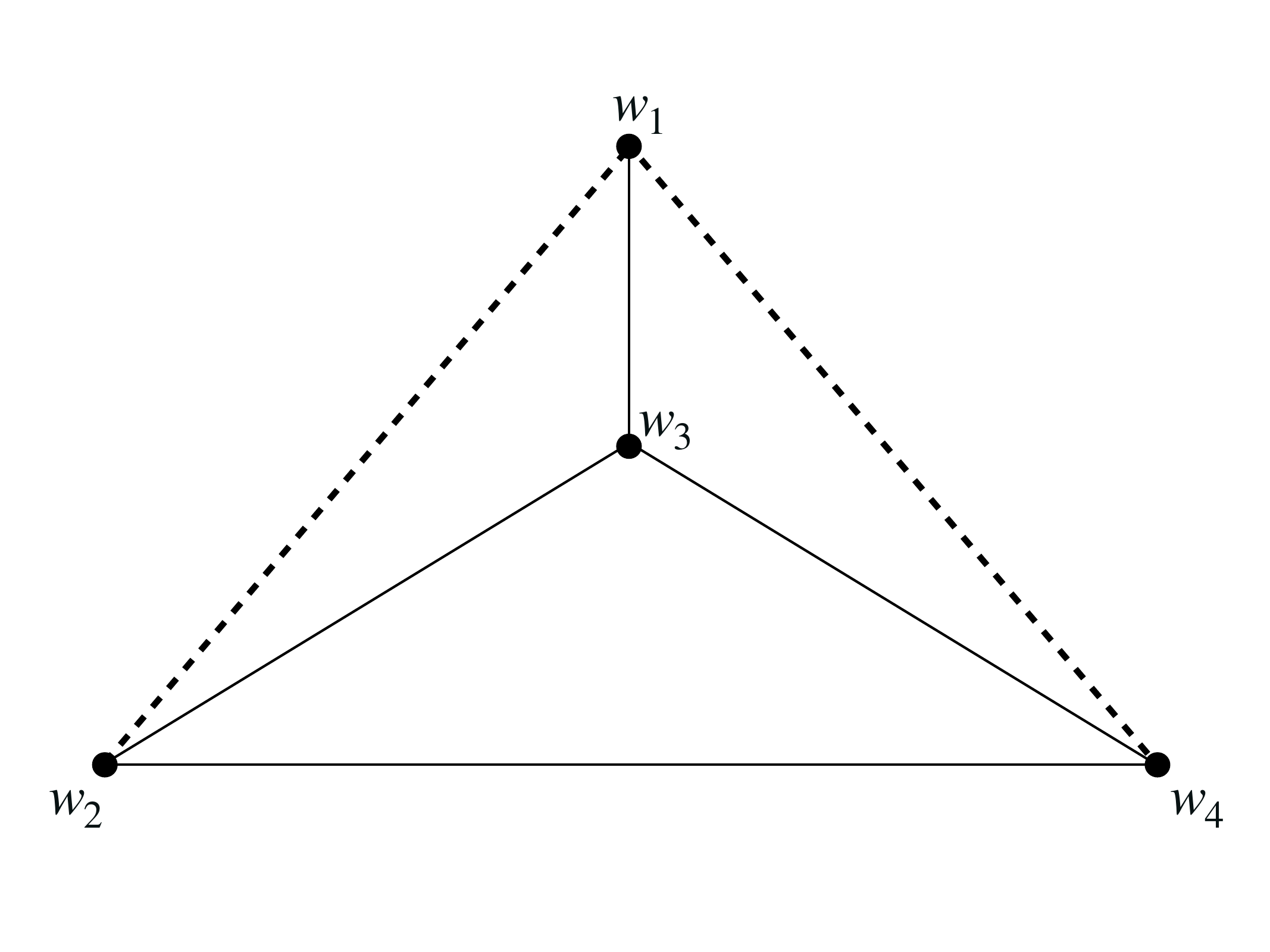}
  \caption{In this gadget graph used in the proof of step one, the dashed edge denotes the nested Clifton graph. As explained in the text, this gadget ensures the property that when $f(w_1) = 1$ it holds that $0 < f(w_2) < 1$ and $0 < f(w_4) < 1$.} 
  \label{fig:stepone-1}
\end{minipage}%
\begin{minipage}{.45\textwidth}
  \centering
  \includegraphics[width=.9\linewidth]{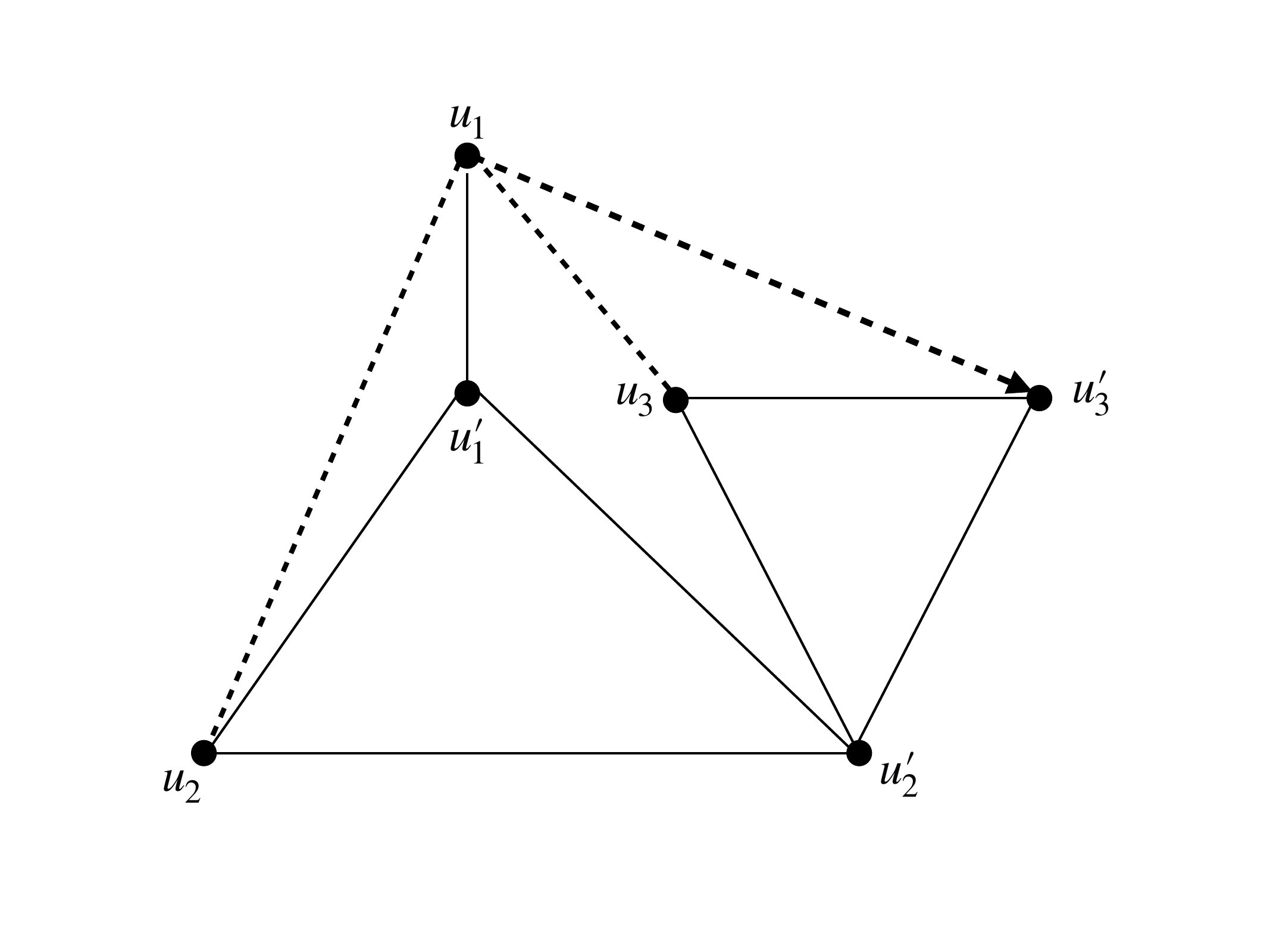}
  \caption{In this gadget graph used in the proof of step one, the directed dashed edge denotes the gadget from Fig. \ref{fig:stepone-1}. As explained in the text, this gadget ensures the property that when $f(u_1) = 1$ it holds that $f(u_2) \leq w_q$ and $f(u_3) \leq w_{q-1}$ in any $O^*$-valued outcome assignment.}
  \label{fig:stepone-2}
\end{minipage}
\end{figure}

We remark that step one of the proof can also be accomplished in a different manner as shown by Hrushovski and Pitowsky in \cite{HP04}. Specifically Lemma \ref{lem:stepone} can be obtained from their Corollary 1 (through a much larger set of vectors and a more involved proof) which we restate here as a lemma.
\begin{lemma}[Corollary 1 of \cite{HP04}]
\label{lem:Pitowskylem}
    Given a ray $z$ in a Hilbert space $\mathcal{H}$ of finite dimension $d \geq 3$ and an integer $t$, there is a finite set of rays $\Lambda_t(z)$ such that every state $p$ on $\Lambda_t(z)$ for which $p(z) = 1$ has at least $t$ distinct values.
\end{lemma}
Choosing $t = q+3$ and the ray $z$ to correspond to vector $| u \rangle$ provides an alternative proof for step one, namely when $f(u) = 1$ the gadget $\Lambda_{q+3}(u)$ in Lemma \ref{lem:Pitowskylem} requires at least $q+3$ distinct values and is hence uncolorable in $O^* = \{0, w_1, \ldots, w_q, 1\}$ (which only allows assignments from $q+2$ distinct values and convex combinations thereof). \\

\textit{Step two:} The second step is proven using the gadget corresponding to the orthogonality graph of the vector set $U_2 = \{|u_1\rangle, \ldots, |u_9 \rangle\}$ explicitly given as:
\begin{align}
    |u_1 \rangle &= \begin{pmatrix} 1 \\ 0 \\ 0 \end{pmatrix}, &
    |u_4 \rangle &= \begin{pmatrix} 0 \\ 1 \\ 0 \end{pmatrix}, &
    |u_7 \rangle &= \begin{pmatrix} 0 \\ 0 \\ 1 \end{pmatrix}, \\
    |u_2 \rangle &= \frac{1}{2}\begin{pmatrix} 0 \\ 1 \\ \sqrt{3} \end{pmatrix}, &
    |u_5 \rangle &= \frac{1}{2}\begin{pmatrix} 1 \\ 0 \\ \sqrt{3} \end{pmatrix}, &
    |u_8 \rangle &= \frac{1}{2}\begin{pmatrix} 1 \\ \sqrt{3} \\ 0 \end{pmatrix}, \\
    |u_3 \rangle &= \frac{1}{2} \begin{pmatrix} 0 \\ -\sqrt{3} \\ 1 \end{pmatrix}, &
    |u_6 \rangle &= \frac{1}{2} \begin{pmatrix} -\sqrt{3} \\0 \\ 1 \end{pmatrix}, &
    |u_9 \rangle &= \frac{1}{2} \begin{pmatrix} -\sqrt{3} \\ 1 \\ 0 \end{pmatrix}.
\end{align}
We first observe that the following triples form complete bases: $B_1 = \{|u_1\rangle, |u_4\rangle, |u_7\rangle\}$, $B_2 = \{|u_1\rangle, |u_2\rangle, |u_3\rangle\}$, $B_3 = \{|u_4\rangle, |u_5\rangle, |u_6\rangle\}$, and $B_4 = \{|u_7\rangle, |u_8\rangle, |u_9\rangle\}$. Now consider any outcome probability assignment $f: U_2 \to O^* \setminus \{1\}$. To obey the KS completeness condition that the sum of probabilities assigned to a complete basis is one, it follows that we must have three linearly independent vectors in the set that are assigned value $0$. Formally we have the following lemma.
\begin{lemma}
    Let $f: U_2 \to O^* \setminus \{1\}$ be a valid outcome probability assignment satisfying the Kochen-Specker completeness condition $\sum_{|u\rangle \in B_k} f(|u\rangle) = 1$ for all $k \in \{1, 2, 3, 4\}$. It then holds that the zero-set $Z = \{|u_i\rangle \in U_2\mid f(|u_i\rangle) = 0\}$ contains exactly three elements, and these three vectors are linearly independent.
\end{lemma}

\begin{proof}
For any complete orthonormal basis $B_k = \{|v_1\rangle, |v_2\rangle, |v_3\rangle\}$, the completeness condition requires $\sum_{i=1}^3 f(|v_i\rangle) = 1$. Consequently, the triple $(f(|v_1\rangle), f(|v_2\rangle), f(|v_3\rangle))$ belongs to $(O^*)^3 \cap \Delta^2$. By Def. ~\ref{def:Ostar}, this triple must lie on the boundary $\partial \Delta^2$, which implies that at least one coordinate must equal zero:
\begin{equation}
    \exists |v_i\rangle \in B_k \quad \text{such that} \quad f(|v_i\rangle) = 0.
\end{equation}
Furthermore, since the codomain of $f$ explicitly excludes one, a basis cannot contain more than one zero, so that every complete basis $B_k$ must contain exactly one vector assigned to $0$. 

Let $Z = \{|u_i\rangle \in \mathcal{U} \mid f(|u_i\rangle) = 0\}$. Consider the first basis $B_1 = \{|u_1\rangle, |u_4\rangle, |u_7\rangle\}$. We have that exactly one element of $B_1$ belongs to $Z$, wlog say $|u_1 \rangle \in Z$. 
This gives that:
\begin{equation}
    f(|u_1\rangle) = 0 \implies f(|u_4\rangle) > 0 \quad \text{and} \quad f(|u_7\rangle) > 0.
\end{equation}
For basis $B_2$, since $|u_1\rangle \in Z$, the unique zero for this basis is already satisfied. For basis $B_3$, since $f(|u_4\rangle) > 0$, we have that $|u_A\rangle \in Z$ where $|u_A\rangle \in \{|u_5\rangle, |u_6\rangle\}$. For basis $B_4$, since $f(|u_7\rangle) > 0$, we have that $|u_B\rangle \in Z$ where $|u_B\rangle \in \{|u_8\rangle, |u_9\rangle\}$. Therefore, $Z = \{|u_1\rangle, |u_A\rangle, |u_B\rangle\}$. By inspection we see that the vectors are linearly independent and span $\mathbb{C}^3$. 

By symmetry, we see that the zero-set $Z$ must belong to one of three configurations:
\begin{equation}
    Z \in \begin{cases}
        \{|u_1\rangle\} \times \{|u_5\rangle, |u_6\rangle\} \times \{|u_8\rangle, |u_9\rangle\}, \\
        \{|u_4\rangle\} \times \{|u_2\rangle, |u_3\rangle\} \times \{|u_8\rangle, |u_9\rangle\}, \\
        \{|u_7\rangle\} \times \{|u_2\rangle, |u_3\rangle\} \times \{|u_5\rangle, |u_6\rangle\}.
    \end{cases}
\end{equation}
\end{proof}

\textit{Step three:} The step three is strictly analogous to the result from \cite{Ravi24} so we do not repeat it here but only outline the key steps. This part of the proof makes use of the gadget shown in Fig. \ref{fig:step3-gadget}. 

\begin{figure}[t]
    \centering
    \includegraphics[width=0.6\textwidth]{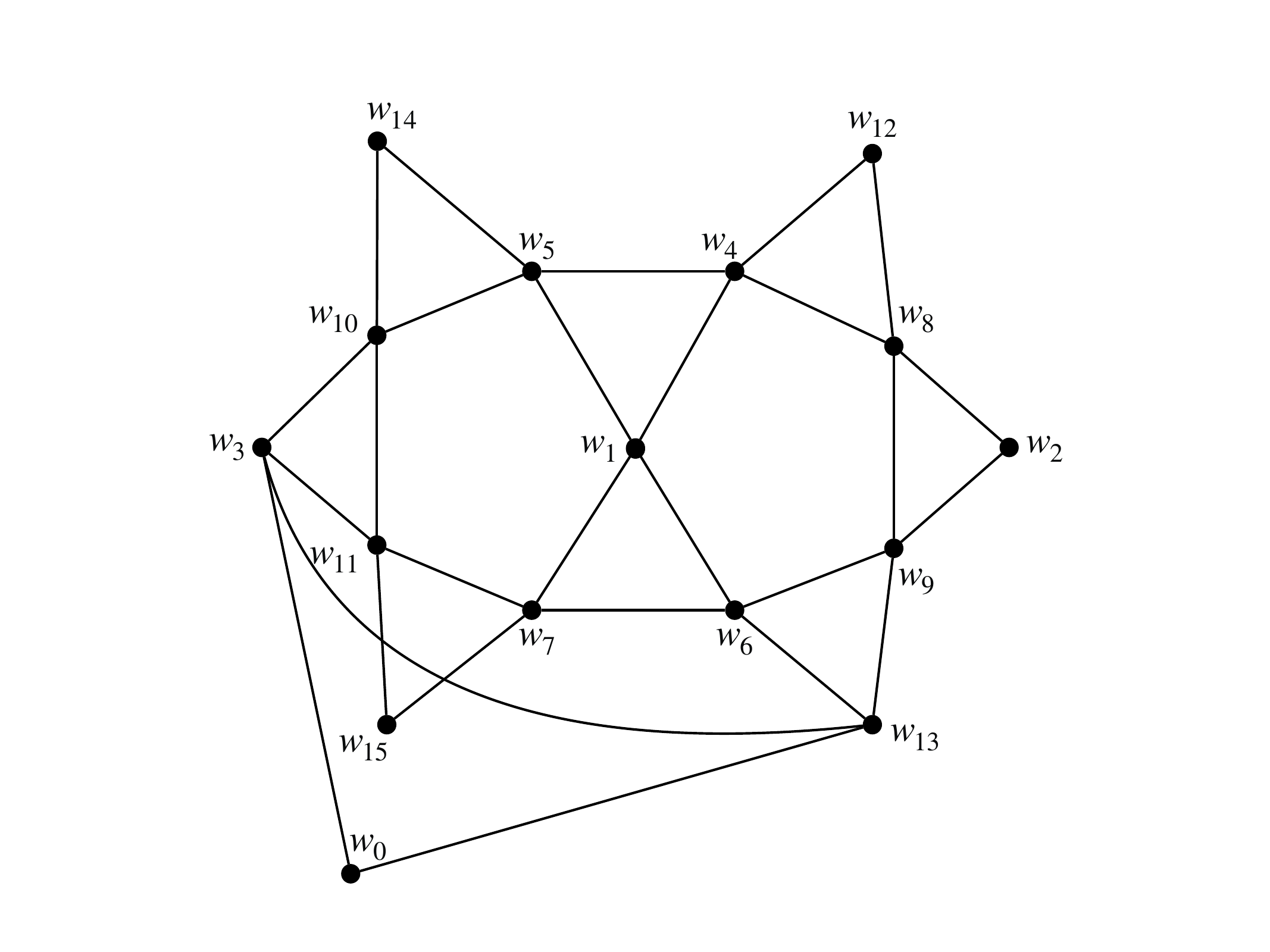}
    \caption{The gadget graph used to establish step three of the proof. Namely that when $f(w_1) = f(w_2) = f(w_3) = 0$, then under the KS constraints of exclusivity and normalisation, it holds that $f(w_{12}) = f(w_{13}) = f(w_{14}) = f(w_{15}) = 0$ giving $f(w_0) = 1$.}
    \label{fig:step3-gadget}
\end{figure}

Specifically consider the set of vectors given as 
\begin{align*} 
|w_1 \rangle &= \begin{pmatrix} 1 \\ 0 \\ 0 \end{pmatrix} & |w_2 \rangle &= \frac{1}{2} \begin{pmatrix} 1 \\ \sqrt{3} \\ 0 \end{pmatrix} & |w_3  \rangle &= \frac{1}{2} \begin{pmatrix} 1 \\ 0 \\ \sqrt{3} \end{pmatrix} & |\tilde{w}_0 \rangle &= \begin{pmatrix}
    - 3 \sqrt{3} \sin^2{(\phi_2)} \\ 4  \\ 3 \sin^2{(\phi_2)}
\end{pmatrix} & \\
|w_4 \rangle &= \begin{pmatrix} 0 \\ \cos{(\phi_1)} \\ \sin{(\phi_1)} \end{pmatrix} & |w_5 \rangle &= \begin{pmatrix} 0 \\ - \sin{(\phi_1)} \\ \cos{(\phi_1)} \end{pmatrix} & |w_6 \rangle &= \begin{pmatrix} 0 \\ \cos{(\phi_2)}
 \\ \sin{(\phi_2)} \end{pmatrix} & |w_7 \rangle &= \begin{pmatrix} 0 \\ -\sin{(\phi_2)} \\ \cos{(\phi_2)} \end{pmatrix} & \\ 
|\tilde{w}_8 \rangle &=  \begin{pmatrix} \sqrt{3} \sin{(\phi_1)} \\ - \sin{(\phi_1)} \\ \cos{(\phi_1)}
     \end{pmatrix} & |\tilde{w}_9 \rangle &= \begin{pmatrix}
         \sqrt{3} \sin{(\phi_2)} \\ -\sin{(\phi_2)} \\ \cos{(\phi_2)} 
     \end{pmatrix} & |\tilde{w}_{10} \rangle &=  \begin{pmatrix}
         \sqrt{3} \sin{(\phi_1)} \\ - \cos{(\phi_1)} \\ -\sin{(\phi_1)} 
     \end{pmatrix} & |\tilde{w}_{11} \rangle &=  \begin{pmatrix}
         \sqrt{3} \sin{(\phi_2)} \\ - \cos{(\phi_2)} \\ - \sin{(\phi_2)} \end{pmatrix} & \\
 |\tilde{w}_{12} \rangle &= \begin{pmatrix}
     1 \\ \sqrt{3} \sin^2{(\phi_1)} \\ - \sqrt{3} \sin{(\phi_1)} \cos{(\phi_1)} 
 \end{pmatrix} & |\tilde{w}_{13} \rangle &= \begin{pmatrix}
     \sqrt{3} \\ 3 \sin^2{(\phi_2)} \\ -1 
 \end{pmatrix} &  |\tilde{w}_{14} \rangle &= \begin{pmatrix}
     1 \\ \sqrt{3} \sin{(\phi_1)} \cos{(\phi_1)} \\ \sqrt{3} \sin^2{(\phi_1)}  \end{pmatrix}  & |\tilde{w}_{15} \rangle &= \begin{pmatrix}
         \sqrt{3} \\ 1 \\ 3 \sin^2{(\phi_2)} 
     \end{pmatrix} &   
\end{align*}
with $\phi_2 = \frac{1}{2}\arcsin{(\frac{2}{3})}$ and $\phi_1 = \arctan{(-\frac{3+\sqrt{5}}{8})}$ so that $\tan(\phi_1)\tan(\phi_2) = -\frac{1}{4}$. Here $|\tilde{w}_j \rangle$ denotes that the vector is not normalized.
We recognize the first row as the vectors $|w_1 \rangle = |u_1 \rangle, |w_2 \rangle = |u_8\rangle, |w_3 \rangle = |u_5\rangle$ from step two, this is one of the twelve triples forming the zero-set $Z$ of step two. We also verify that the vectors faithfully realize the orthogonality graph shown in Fig. \ref{fig:step3-gadget}. 

Now consider the constraint from step two, namely that $f(|w_1 \rangle) = f(|w_2 \rangle) = f(|w_3 \rangle) = 0$. To ensure the normalization conditions we obtain that $f(|w_4 \rangle) + f(|w_5 \rangle) = 1$, $f(|w_6 \rangle) + f(|w_7 \rangle) = 1$, $f(|w_8 \rangle) + f(|w_9 \rangle) = 1$ and $f(|w_{10} \rangle) + f(|w_{11} \rangle) = 1$. On the other hand, we also have $f(|w_4 \rangle) + f(|w_8 \rangle) + f(|w_{12} \rangle) = 1$, $f(|w_5 \rangle) + f(|w_{10} \rangle) + f(|w_{14} \rangle) = 1$, $f(|w_7 \rangle) + f(|w_{11} \rangle) + f(|w_{15} \rangle) = 1$ and $f(|w_6 \rangle) + f(|w_9 \rangle) + f(|w_{13} \rangle) = 1$. These two conditions together imply that $f(|w_{12} \rangle) = f(|w_{13} \rangle) = f(|w_{14} \rangle) = f(|w_{15} \rangle) = 0$. Finally, for the basis $\{|w_0 \rangle, |w_{3} \rangle, |w_{13} \rangle \}$ we see that with $f(|w_3 \rangle) = 0$ and $f(|w_{13} \rangle) = 0$ it must be the case that $f(|w_0 \rangle) = 1$. Augmenting the graph with a gadget from step one which ensures $f(|w_0 \rangle) \neq 1$ we obtain a contradiction. Therefore, as promised we have constructed a gadget that admits assignments from $O^* \setminus \{1\}$ and such that in any such assignment it is necessarily the case that the three linearly independent vectors $|w_1 \rangle, |w_2 \rangle, |w_3 \rangle$ from step two cannot all be assigned value $0$. In \cite{Ravi24}, the final step is shown, i.e., constructing the gadget for which it is necessarily the case that any of the triples in the zero set $Z$ cannot all be assigned value $0$. Taking the union of the (finite) gadgets from steps one, two and three gives the proof.

\end{proof}

    
    


\begin{thebibliography}{99}
\bibitem{KS67} S. Kochen, and E.P. Specker. \textit{The problem of hidden variables in quantum mechanics}. Journal of Mathematics and Mechanics, 17(1), 59-87 (1967).
\bibitem{JXSPC18} S.-H. Jiang, Z.-P. Xu, H.-Y. Su, A. K. Pati, and J.-L. Chen. \textit{Generalized Hardy’s Paradox}. Phys. Rev. Lett. 120, 050403 (2018). 
\bibitem{Mermin94} N. D. Mermin. \textit{Quantum mysteries refined}. Am. J. Phys. 62, 880–887 (1994). 
\bibitem{CBCB13} A. Cabello, P. Badziag, M. T. Cunha and M. Bourennane. \textit{Simple Hardy-Like Proof of Quantum Contextuality}. Phys. Rev. Lett. 111, 180404 (2013). 
\bibitem{MANCB14} B. Marques, J. Ahrens, M. Nawareg, A. Cabello, and M. Bourennane. \textit{Experimental observation of Hardy-like quantum contextuality}. Phys. Rev. Lett. 113, 250403 (2014).








\bibitem{HWVE14} M. Howard, J. J. Wallman, V. Veitch, J. Emerson \textit{Contextuality supplies the magic for quantum computation}. Nature 510, 351 (2014). 
\bibitem{BRV+19} K. Bharti, M. Ray, A. Varvitsiotis, N. A. Warsi1, A. Cabello, and L.-C. Kwek. \textit{Robust Self-Testing of Quantum Systems via Noncontextuality Inequalities}. Phys. Rev. Lett. 122, 250403 (2019). 
\bibitem{Gleason57} A. M. Gleason. \textit{Measures on the Closed Subspaces of a Hilbert Space}. Journal of Mathematics and Mechanics, 6, 885–893 (1957).
\bibitem{CKM85} R. Cooke, M. Keane, and W. Moran. \textit{An elementary proof of Gleason's theorem}. Math. Proc. Camb. Philos. Soc. 98, 117 (1985).
\bibitem{Pitowsky98} I. Pitowsky. \textit{Infinite and Finite Gleason’s Theorems and the Logic of Indeterminacy}. Journal of Mathematical Physics 39, 218 - 228 (1998).
\bibitem{Pitowsky03} I. Pitowsky. \textit{Betting on the Outcomes of Measurements: A Bayesian Theory of Quantum Probability}. Studies in the History and Philosophy of Modern
Physics 34, 395-414 (2003). 
\bibitem{RB99} F. Richman and D. Bridges. \textit{A Constructive Proof of Gleason's Theorem}. Journal of Functional Analysis 162, 287-312 (1999).
\bibitem{Fine82} A. Fine. \textit{Joint distributions, quantum correlations, and commuting observables}. J. Math. Phys. 23, 1306–1310 (1982). 
\bibitem{Vorobev62} N. N. Vorobev. \textit{Theory of Probability and Its Applications}. 7, 147-163 (1962).
\bibitem{CSW14} A. Cabello, S. Severini, and A. Winter. \textit{Graph Theoretic Approach to Quantum Correlations}.
Phys. Rev. Lett. 112, 040401 (2014).
\bibitem{LR25} Y. Liu and R. Ramanathan. \textit{Optimal and feasible contextuality-based randomness generation}. Phys. Rev. Lett. 135 (17), 170805 (2025).
\bibitem{LRH+23} Y. Liu et al. \textit{Optimal measurement structures for contextuality applications}. npj Quantum Information 9(1), 63 (2023).
\bibitem{Ravi24} R. Ramanathan. \textit{Generalised Kochen-Specker theorem for finite non-deterministic outcome assignments}. npj Quantum Information 10(1), 99 (2024). 
\bibitem{WC25} T. Williams and A. Constantin. \textit{Maximal non-Kochen-Specker sets and a lower bound on the size of Kochen-Specker sets}. Phys. Rev. A 111, 012223 (2025). 
\bibitem{KS15} R. Kunjwal and R. W. Spekkens. \textit{From the Kochen-Specker theorem to noncontextuality inequalities without assuming determinism}. Phys. Rev. Lett. 115, 110403 (2015). 
\bibitem{BvN36} G. Birkhoff, and J. von Neumann. \textit{The logic of quantum mechanics}. Annals of Mathematics, 37(4), 823-843 (1936).
\bibitem{Reichenbach44} H. Reichenbach. \textit{Philosophic Foundations of Quantum Mechanics}. Mineola, N.Y.: Dover Publications (1944).
\bibitem{Lukasiewicz70} J. {\L}ukasiewicz. Selected Works. In: Borkowski. L (ed). North-Holland, Amsterdam (1970). 
\bibitem{Pykacz94} J. Pykacz. \textit{Fuzzy quantum logic and infinite-valued {\L}ukasiewicz logic}. International Journal of Theoretical Physics, 33, 1403-1416 (1994). 
\bibitem{Maczynski73} M. J. Maczy\'{n}ski. \textit{The orthogonality postulate in axiomatic quantum mechanics}. Int. J. Theor. Phys. 8, 353–360 (1973).
\bibitem{KGP+11} M. Kleinmann, O. G\"{u}hne, J. R. Portillo, J. A. Larsson, A. Cabello. \textit{Memory cost of quantum contextuality} New J. Phys. 13, 113011 (2011).
\bibitem{Hardy04} L. Hardy. \textit{Quantum ontological excess baggage}. Studies in History and Philosophy of Science Part B: Studies in History and Philosophy of Modern Physics, 35(2), 267-276 (2004).
\bibitem{Montina08} A. Montina. \textit{Exponential complexity and ontological theories of quantum mechanics}. Phys. Rev. A 77, 022104 (2008).
\bibitem{Pavicic92} M. Pavicic. \textit{Bibliography on quantum logics and related structures}. International Journal of Theoretical Physics 31, 373–461 (1992). 
\bibitem{KDL15} J. V. Kujala, E. N. Dzhafarov, J.-Å. Larsson. \textit{Necessary and Sufficient Conditions for Extended Noncontextuality in a Broad Class of Quantum Mechanical Systems}. Phys. Rev. Lett. 115, 150401 (2015). 
\bibitem{Fyrillas23} A. Fyrillas et al. \text{Certified randomness in tight space}. PRX Quantum 5, 020348 (2024).
\bibitem{AB11} S. Abramsky, A. Brandenburger. \textit{The Sheaf-Theoretic Structure Of Non-Locality and Contextuality}. New Journal of Physics 13, 113036 (2011). 
\bibitem{ABKLM15} S. Abramsky, R. S. Barbosa, K. Kishida, R. Lal, S. Mansfield. \textit{Contextuality, Cohomology and Paradox}.	24th EACSL Annual Conference on Computer Science Logic (CSL 2015), Leibniz International Proceedings in Informatics (LIPIcs), 41: 211-228 (2015). 
\bibitem{HZS+23} K. Horodecki et al. \textit{The rank of contextuality}. New Journal of Physics 25(7), 073003 (2023).
\bibitem{Liu2024} Y. Liu et al. \textit{Equivalence between face nonsignaling correlations, full nonlocality, all-versus-nothing proofs, and pseudotelepathy}. Phys. Rev. Research 6, L042035 (2024). 
\bibitem{GGM74} E. R. Gerelle, R. J. Greechie, and F. R. Miller. \textit{Weights on spaces}. In Physical Reality and Mathematical Description, edited by C. P. Enz and J. Mehra (D. Reidel Publishing Company, Springer Netherlands, Dordrecht, The
Netherlands, pp. 167–192 (1974).
\bibitem{Wright78} R. Wright. \textit{The state of the pentagon. A nonclassical example}. In Mathematical Foundations of Quantum Theory, edited by A. R. Marlow (Academic Press, New York, pp. 255–274 (1978).
\bibitem{AQBCC13} M. Ara\'{u}jo, M. T. Quintino, C. Budroni, M. T. Cunha, and A. Cabello. \textit{All noncontextuality inequalities for the n-cycle scenario}. Phys. Rev. A 88, 022118 (2013). 
\bibitem{HP04} E. Hrushovski and I. Pitowsky. \textit{Generalizations of Kochen and Specker's theorem and the effectiveness of Gleason's theorem}. Studies in History and Philosophy of Science Part B: Studies in History and Philosophy of Modern Physics Vol. 35, Issue 2, 177-194 (2004).
\bibitem{Svozil98} K. Svozil. \textit{Quantum logic}. 998/9/1
Springer Science \& Business Media (1998).
\bibitem{JM05} N. S. Jones and Ll. Masanes. \textit{Interconversion of nonlocal correlations}. Phys. Rev. A 72, 052312 (2005).
\bibitem{BP05} J. Barrett and S. Pironio. \textit{Popescu-Rohrlich correlations as a unit of nonlocality}. Phys. Rev. Lett. 95, 140401 (2005).
\bibitem{BLMPPR05} J. Barrett, N. Linden, S. Massar, S. Pironio, S. Popescu and D. Roberts. \textit{Non-local correlations as an information theoretic resource}. Phys. Rev A 71, 022101 (2005).
\bibitem{PR94} S. Popescu and D. Rohrlich. \textit{Quantum nonlocality as an axiom}. Found Phys 24, 379-385 (1994).
\bibitem{Clifton93} R. K. Clifton. \textit{Getting Contextual and Nonlocal Elements-of-Reality the Easy Way}. American Journal of Physics, 61: 443 (1993).
\bibitem{Stairs83} A. Stairs. \textit{Quantum logic, realism, and value definiteness}. Philosophy of Science 50(4), 578-602 (1983).
\bibitem{Belinfante73} F. J. Belinfante. A survey of hidden-variables theories. Elsevier, 55 (1973).
\bibitem{LSS87}  L. Lovasz, M. Saks, A. Schrijver. \textit{Orthogonal representation
and connectivity of graphs}. Linear Algebra and its applications, 4, 114-115, 439 (1987).
\bibitem{Piron76} C. Piron. \textit{Foundations of Quantum Physics}. Addison-Wesley, (1976).
\bibitem{Mackey63} G. W. Mackey. \textit{Mathematical Foundations of Quantum Mechanics}. Addison-Wesley, (1963).
\bibitem{RW04} R. Renner and S. Wolf. \textit{Quantum pseudo-telepathy and the Kochen-Specker theorem}. Int. Symp. Information Theory (ISIT), Proc. (Piscataway, NJ:IEEE) p322 (2004).
\bibitem{RRHS+20} R. Ramanathan, M. Rosicka, K. Horodecki, S. Pironio, M. Horodecki, and P. Horodecki. \textit{Gadget structures in proofs of the Kochen-Specker theorem}. Quantum 4, 308 (2020). 
\bibitem{Busch03} P. Busch. \textit{Quantum states and generalized observables: a simple proof of Gleason's theorem}. Phys. Rev. Lett. 91, 120403 (2003).
\bibitem{Soler95} M. P. Sol\'{e}r. \textit{Characterization of Hilbert spaces by orthomodular spaces}. Communications in Algebra 23, 219-243 (1995). 
\bibitem{Enderton01}  H. B. Enderton. \textit{A Mathematical Introduction to Logic} (2nd ed.), Academic Press (2001).
\bibitem{RSKK12}  R. Ramanathan, A. Soeda, P. Kurzynski, and D. Kaszlikowski. \textit{Generalized monogamy of contextual inequalities from the no-disturbance principle}. Phys. Rev. Lett. 109, 050404 (2012). 
\bibitem{AFLS15} A. Ac\'{i}n, T. Fritz, A. Leverrier, and A. B. Sainz. \textit{A Combinatorial Approach to Nonlocality and
Contextuality}. Comm. Math. Phys. 334(2), 533-
628 (2015).
\bibitem{ABM17} S. Abramsky, R. S. Barbosa, and S. Mansfield. \textit{Contextual Fraction as a Measure of Contextuality}. Phys. Rev. Lett. 119, 050504 (2017). 
\bibitem{GM17} T. D. Galley and Ll. Masanes. \textit{Classification of all alternatives to the Born rule in terms of informational properties}. Quantum 1, 15 (2017).
\bibitem{Aaronson04} S. Aaronson. \textit{Is Quantum Mechanics An Island In Theoryspace?}  arXiv:0401062 (2004).
\bibitem{our-4} R. Ramanathan, M. Horodecki, H. Anwer, S. Pironio, K. Horodecki, M. Gr\"{u}nfeld, S. Muhammad, M. Bourennane and P. Horodecki. \textit{Practical No-Signalling proof Randomness Amplification using Hardy paradoxes and its experimental implementation}. arXiv:1810.11648 (2018). 
\bibitem{BarrettPRL} J. Barrett, L. Hardy and A. Kent. \textit{No signaling and quantum key distribution}. Phys. Rev. Lett. \textbf{95}, 010503 (2005).
\bibitem{ZRLH22} S. Zhao, R. Ramanathan, Y. Liu, P. Horodecki. \textit{Tilted Hardy paradoxes for device-independent randomness extraction}. Quantum 7, 1114 (2023).
\bibitem{SK25} V. K. Sujan, R. Kunjwal. \textit{No-go theorem for quantum realization of extremal correlations}. arXiv:2509.14879 (2025). 
\bibitem{RTHH16} R. Ramanathan, J. Tuziemski, M. Horodecki, and P. Horodecki. \textit{No Quantum Realization of Extremal No-Signaling Boxes}. Phys. Rev. Lett. 117, 050401 (2016). 
\bibitem{Lovasz06} L. Lov\'{a}sz. \textit{Graph minor theory}. Bulletin of the American Mathematical Society, 43 (1): 75–86 (2006).
\bibitem{Diestel05} R. Diestel. \textit{Graph Theory} (3rd ed.), Springer,(2005). 


















\end{thebibliography}
\end{document}